\documentclass[11pt,a4paper,oneside]{article}
\usepackage[top=3cm, bottom=3cm, left=2cm, right=2cm]{geometry}
\usepackage[english]{babel}
\usepackage[utf8]{inputenc}
\usepackage{enumerate}
\usepackage[T1]{fontenc}
\usepackage{lmodern}
\usepackage{setspace}
\usepackage{xcolor}
\usepackage{booktabs}
\usepackage{multirow}
\usepackage{xr}
\usepackage{tikz}
\usetikzlibrary{positioning}
\usetikzlibrary{fit}

\usepackage{amsthm,amsmath,amssymb,mathrsfs,dsfont,mathtools, mathabx}
\usepackage{bbm}
\usepackage{bm}
\usepackage{graphicx}
\usepackage[colorlinks=true, allcolors=blue]{hyperref}
\usepackage{comment}
\usepackage[all]{nowidow}

\usepackage[authoryear, round]{natbib}
\usepackage{authblk}
\date{ \vspace{-5mm} }
\title{Product Centered Dirichlet Processes for Bayesian Multiview Clustering}
\author[1]{Alexander Dombowsky}
\author[1,2]{David B. Dunson}
\affil[1]{Department of Statistical Science, Duke University, Durham, NC, USA}
\affil[2]{Department of Mathematics, Duke University, Durham, NC, USA}

\usepackage[sf]{titlesec}
\usepackage{sectsty}
\allsectionsfont{\centering \normalfont\scshape} 

\theoremstyle{definition}

\newtheorem{theorem}{Theorem}

\newtheorem{remark}{Remark}
\newtheorem{prop}{Proposition}

\def\bs{\boldsymbol}
\def\lb{\left\{ }
\def\rb{\right\}}

\def\E{\mathbb{E}}
\def\Var{\textrm{Var}}

\def\tx{\textrm}
\def\X{\bs X}

\def\diff{\tx{d}}

\def\KV{K_v}
\def\n{\bs n}
\def\CLIC{\tx{CLIC}}
\def\CRP{\tx{CRP}}
\def\Dir{\tx{Dir}}

\def\Pn{\mathcal{P}_n}
\def\N{\mathcal{N}}

\def\DP{\tx{DP}}
\def\cvp{\overset{p}{\longrightarrow}}
\def\CRP{\tx{CRP}}

\newcommand{\ind}{\perp\!\!\!\!\perp} 
\newcommand{\norm}[1]{\left\lVert#1\right\rVert}

\begin{document}

\maketitle

\begin{abstract}
While there is an immense literature on Bayesian methods for clustering, the multiview case has received little attention. This problem focuses on obtaining distinct but statistically dependent clusterings in a common set of entities for different data types. For example, clustering patients into subgroups with subgroup membership varying according to the domain of the patient variables. A challenge is how to model the across-view dependence between the partitions of patients into subgroups. 
The complexities of the partition space 
make standard methods to model dependence, such as correlation, infeasible. In this article, we propose CLustering with Independence Centering (CLIC), a clustering prior that uses a single parameter to explicitly model dependence between clusterings across views. CLIC is induced by the product centered Dirichlet process (PCDP), a novel hierarchical prior that bridges between independent and equivalent partitions. We show appealing theoretic properties, provide a finite approximation and prove its accuracy, present a marginal Gibbs sampler for posterior computation, and derive closed form expressions for the marginal and joint partition distributions for the CLIC model. On synthetic data and in an application to epidemiology, CLIC accurately characterizes view-specific partitions while providing inference on the dependence level.
\end{abstract}

\doublespacing

\section{Introduction} \label{section:intro}
Many application areas collect {\em multiview} data comprising several different types of information on the same set of objects. Multiomics data, for example, consist of various molecular features, such as single nucleotide polymorphisms (SNPs), 
RNA expression, and DNA methylation measured on the same sample \citep{rappoport2018multi}.  The particular focus of this article is on {\em multiview clustering}, a field which generalizes standard cluster analysis to multiview data; see \cite{bickel2004multi} and \cite{fang2023comprehensive} for an overview of the developments in this area. For example, it is broadly of interest to infer patient subgroups based on three views: molecular biomarkers, clinical covariates, and disease outcomes. These views may be of vastly different dimension or type. Often, the biomarkers measure counts for tens of thousands of genes, the covariates consist of categorical demographic information and numeric laboratory results, and the disease outcomes are a single binary variable indicating mortality. Multiview clustering methods seek to incorporate or balance these markedly different attributes when inferring the subgroups.

Previously, the term multiview clustering referred to a technique in which contributions from the different views are incorporated into a single integrative clustering. However, more recent approaches have expanded the notion of multiview clustering to address the aim of {\em multiple clustering} \citep{yao2019multivew, wei2020multi, franzolini2023conditional}, a related field which derives multiple clusterings for a single dataset \citep{yu2024multiple}. While the input for a multiple clustering algorithm need not be multiview, the clusterings generated by said algorithm are created to be as dissimilar as possible to reflect different perspectives of the data \citep{bailey2014alternative}. In this article, we refer to multiview clustering as the problem of inferring separate but statistically dependent groupings in the same objects for the different views under consideration. Returning to our earlier example, this framework entails inferring three groupings of the hospital patients: one informed by the molecular biomarkers, another based on clinical characteristics, and a third reflecting differences in the conditional distribution of the outcomes given the biomarkers and covariates. 

There are several practical reasons why multiview clustering is preferable to applying a clustering algorithm to a single, consolidated feature set. A key benefit of multiview clustering is interpretability of the results; the clustering associated with a certain view reflects the variation within that view. Next, the multiview approach ensures that some views do not become obscured in the clustering procedure by other views of higher dimension. In the case where one view is a low dimensional outcome and the other is a set of covariates, the posterior distribution under a single partition is dominated by the likelihood of the predictors, obscuring any grouping structure indicated by the dependent variable \citep{wade2014improving}. Even in low dimensional settings, clustering algorithms are skewed towards capturing variability in the views with higher dimension \citep{franzolini2023bayesian}. Finally, investigators are often interested in the relationship between the clusterings, e.g., assessing whether covariate subgroups map well to patterns in the outcome.

The Bayesian framework is appealing for multiview clustering due to its ability to model complex dependence structures and quantify clustering uncertainty. However, Bayesian methodology was seldom applied to multiview clustering analyses until relatively recently, and earlier models focused on capturing a specific notion of dependence rather than defining a general measure of clustering correlation. Bayesian consensus clustering \citep{lock2013bayesian,lu2022bayesian} induces across view dependence by a latent integrative clustering of the objects and is primarily motivated by multiomics data. The integrative clustering is not just a vessel to imbue dependence, but moreso a primary focus of inference in order to obtain a single unified grouping of the observations. In general, we may not want to assume that such an integrative clustering exists. Alternatively, enrichment \citep{wade2011enriched, wade2014improving} addresses prediction with clustering models, and enforces dependence by requiring that the clusters in the outcome be subsets of the clusters in the covariates. Nested clustering models have also been utilized by \cite{lee2016nonparametric} and \cite{franzolini2023bayesian}. The nested assumption can be rather rigid in practice, as clusterings can still be similar without having a tree structure. 

We focus on inducing a broader sense of cluster dependence. Generally, there are two extremes for specifying a probabilistic clustering model for multiview data. On one end of the spectrum is the single partition model, which arises when all view-specific partitions are equivalent. In this case, the relationship between the partitions is trivially dependent. Conversely, assuming that there are independent partitions associated with each view fails to acknowledge the possibility of across-view dependence induced by clustering the same set of objects. If two objects are co-clustered together under one view, we would like to incorporate this information when deciding whether to co-cluster them under another.

In this paper, we propose a framework for modeling dependent partitions with a single parameter controlling dependence, bridging between identical partitions for the different views and independence at the two extremes. Let $\Pn$ be the set partitions of the integers $[n] = \{ 1, \dots, n \}$ and suppose we randomly draw two partitions $C_1, C_2 \in \Pn$, where $C_v = (C_{v1}, \dots, C_{vK_v})$. Alternatively, we may express the partitions as $n$-dimensional labelings $\bs c_v = (c_{v1}, \dots, c_{vn})$ with $c_{vi} \in [\KV] = \{ 1, \dots, \KV \}$ for all $i=1, \dots, n$. $C_1$ and $C_2$ represent partitions of the same data for different views. The \textit{Rand index} \citep{rand1971objective} between $C_1$ and $C_2$ is defined as $R(C_1, C_2) = \{ A_{12} + B_{12} \}/{n \choose 2}$, where $A_{12}= \sum_{i<j} \textbf{1}_{c_{1i}=c_{1j}, c_{2i}=c_{2j}}$ and $B_{12} = \sum_{i<j} \textbf{1}_{c_{1i} \neq c_{1j}, c_{2i} \neq c_{2j}}$ count the number of shared pairwise clustering decisions. If $C_1$ and $C_2$ are equivalent, $R(C_1, C_2) = 1$. Assuming that $\pi$ is a joint distribution with support on $\Pn$, we will be interested in the expected Rand index (ERI) $\tau_{12}:= \E[R(C_1, C_2)] = \sum_{C_1, C_2 \in \Pn} R(C_1, C_2) \pi(C_1, C_2)$. If $C_1 \ind C_2$ a priori, i.e. $\pi(C_1, C_2) = \pi(C_1)\pi(C_2)$, we denote the resulting ERI as $\tau_{12}^{0} = \sum_{C_1, C_2 \in \Pn} R(C_1, C_2) \pi(C_1) \pi(C_2)$. The form of $\tau_{12}^{0}$ depends on the hyperparameters of the marginal probability mass functions (PMFs) $\pi(C_1)$ and $\pi(C_2)$. We introduce CLustering with Independence Centring (CLIC), a class of joint distributions on $\Pn$ with the property that for some $\rho>0$,
\begin{equation} \label{eq:CLIC-definition}
    \tau_{12}= \frac{1}{\rho+1} + \left( 1 - \frac{1}{\rho+1} \right) \tau_{12}^{0}.
\end{equation}
That is, the ERI is a weighted average between $1$, the ERI of two equivalent random partitions, and $\tau_{12}^0$, the ERI of two independent random partitions. Notice that $\rho$ controls concentration towards the independent clustering model. In practice, one is uncertain about $\rho$, but we will show that a conjugate prior naturally arises under a specific data augmentation scheme. This will allow one to infer the posterior distributions of both $(C_1,C_2)$ and $\rho$. 

Our model is induced by altering the Dirichlet process (DP) prior \citep{ferguson1973bayesian} for generating random measures with support on some space $\Theta$. The DP is parameterized with a \textit{concentration parameter} $\gamma > 0$, and a \textit{base measure} $H$ on $\Theta$. Let $S_1, \dots, S_K$ be any partition of $\Theta$, i.e $S_k \cap S_l = \emptyset$ and $\bigcup_{k=1}^K S_k = \Theta$.  We say that $P \sim \DP(\gamma, H)$ if $(P(S_1), \dots, P(S_K)) \sim \Dir(\gamma H(S_1), \dots, \gamma H(S_K))$. The elementary properties of the Dirichlet distribution imply that $H$ can be viewed as the center of the DP with variation controlled by $\gamma$. The larger $\gamma$ is, the more concentrated the prior on $P$ is at the base measure $H$. 

\cite{sethuraman1994constructive} showed that realizations from the Dirichlet process are almost surely discrete even if $H$ is non-atomic. For that reason, a sample $\theta_1, \dots, \theta_n \sim P$ is expected to have a number of repeated values. By grouping objects with the same value of $\theta_i$, we create index sets known as \textit{clusters}, ultimately inducing a partition $C = (C_1, \dots, C_K) \in \Pn$. If we marginalize out the DP prior on $P$, the induced prior on $C$ is the Chinese Restaurant Process (CRP) \citep{aldous1985exchangeability}. $C \sim \CRP(\gamma)$ if its PMF is
\begin{equation} \label{eq:CRP-PMF}
    \pi^{\CRP}(C \mid \gamma) = \frac{\gamma^K}{\prod_{i=1}^n (\gamma + i - 1)} \prod_{k=1}^K (n_k -1)!,
\end{equation}
where $n_k = |C_k|$ is the size of the $k$th cluster for $k=1, \dots, K$. The number of clusters in the CRP is governed by $\gamma$, and one can show that $\E(K)/\log n \to \gamma$ as $n \to \infty$. The clustering properties of the DP prior are particularly useful in Dirichlet Process Mixtures (DPMs) in which $\theta_i$ determines the distribution of data for object $i$ \citep{antoniak1974mixtures}. 

We return now to our original aim of modeling dependent partitions. Suppose that $\Theta = \Theta_1 \times \Theta_2$ for some parameter spaces $\Theta_v$ for $v=1,2$. Here, the index $v$ corresponds to the view of the data, whether it be multiomics or covariates and an outcome. To induce two view-specific partitions of $[n]$, we decompose the atoms into $\theta_i = (\theta_{1i}, \theta_{2i})$, where $\theta_{vi} \in \Theta_v$. We define the product centered Dirichlet process (PCDP) to be the following prior for a random measure $P$ on $\Theta$,
\begin{equation} \label{eq:CLIC-model}
    \begin{gathered} 
    P \mid \rho, Q_1, Q_2 \sim \DP(\rho, Q_1 \times Q_2); \\
    Q_v \mid \gamma_v, H_v \sim \DP(\gamma_v, H_v) \tx{ for } v=1,2;
\end{gathered}
\end{equation}
where $\gamma_v > 0$ and $H_v$ has support on $\Theta_v$, and $\times$ denotes the product measure. By counting ties in $\theta_{1i}$ and $\theta_{2i}$, the PCDP creates two dependent partitions $C_1, C_2 \in \Pn$ with ERI given by \eqref{eq:CLIC-definition} and $\tau_{12}^{0} = \sum_{C_1, C_2 \in \Pn} R(C_1, C_2) \pi^\CRP(C_1 \mid \gamma_1) \pi^\CRP(C_2 \mid \gamma_2)$. The induced joint partition distribution is denoted by $(C_1, C_2) \sim \CLIC(\rho, \gamma_1, \gamma_2)$. We will show that increasing values of $\rho$ will result in independent $\CRP(\gamma_v)$ partitions, whereas increasing values of $\gamma_v$ results in two identical $\CRP(\rho)$ partitions.  

The PCDP is related to but distinct from methods for 
Bayesian nonparametric modeling of data from subjects in different groups, including the hierarchical Dirichlet process (HDP) \citep{teh2006sharing, ren2008dynamic, paisley2014nested} and the nested Dirichlet process (NDP) \citep{rodriguez2008nested}.
In this literature, groups are prespecified and contain different subjects, and the focus is on borrowing of information in inferring group-specific distributions while clustering the subjects in different groups. Particularly relevant is the HDP, which models a random measure for each group using a hierarchical DP prior in which the base measure $Q$ is itself modeled as a DP.  
Conversely, \eqref{eq:CLIC-model}  simulates a single random probability distribution whose base measure is a product of DP distributed measures. The goal of the PCDP is fundamentally different in inducing dependence across views in clustering a single set of subjects. However, the marginal distributions of measures drawn from a PCDP will be HDPs, and the Gibbs sampler for the HDP will be useful for inferring the posterior of our induced CLIC model. More generally, the PCDP is not an example of a hierarchical process \citep{camerlenghi2018bayesian, camerlenghi2019distribution}, but theoretical tools established by the literature on hierarchical processes will be essential for deriving the closed form expression of the $\CLIC(\rho, \gamma_1, \gamma_2)$ probability mass function. 

The PCDP and CLIC contribute to an emerging literature on Bayesian inference for dependent clustering models. Besides nested clusterings such as the enriched Dirichlet process, 
temporally correlated partitions have been modeled using a generalized P\'olya urn scheme in \cite{caron2007generalized} and \cite{caron2017generalized}, and a sequence of cluster re-allocations in \cite{page2022dependent}. The latter approach exhibits a similar ERI to \eqref{eq:CLIC-definition} when $n=2$. An approximate Bayes approach for multiview clustering was proposed in \cite{duan2020latent} that models the similarity matrix of each view as a noisy realization from a cluster graph. \cite{franzolini2023bayesian} introduced conditional partial exchangeability (CPE), a general framework for modeling dependent partitions that encompasses the models of \cite{wade2011enriched}, \cite{lee2016nonparametric}, and \cite{page2022dependent}. A two view clustering model satisfies CPE if the data are exchangeable conditional on $C_1$ and $C_2$, and, additionally, the data in the second view are partially exchangeable conditional on $C_1$. CLIC satisfies these CPE assumptions under certain settings. 
Furthermore, the CLIC 
parameter $\rho$ in \eqref{eq:CLIC-definition} is a transformation of the expected telescopic adjusted Rand index (TARI), a measure of clustering dependence proposed by \cite{franzolini2023bayesian}. In contrast to their methods, CLIC directly models uncertainty in clustering dependence by assigning $\rho$ a prior, does not assume a directional relationship between $C_1$ and $C_2$, and induces a different joint partition probability function for $(C_1,C_2)$.

This paper is organized as follows. In Section \ref{section:PCDP} we propose the PCDP as a model for dependent random measures, show basic theoretic properties, and derive a finite approximation that is analogous to other random measure priors \citep{ishwaran2002exact, lijoi2023finite}. In Section \ref{section:CLIC}, we focus on CLIC, the induced dependent partition distribution of the PCDP. We prove that \eqref{eq:CLIC-definition} holds for the $\CLIC(\rho, \gamma_1, \gamma_2)$ distribution and derive the closed form expression for its PMF, which we call a marginally exchangeable partition probability function (MEPPF). A marginal Gibbs sampler for the finite approximation is discussed in Section \ref{section:computation}. We evaluate CLIC as a clustering method for synthetic multiview data in Section \ref{section:simulations} and in an application to the Collaborative Perinatal Project (CPP) in Section \ref{section:application}. Finally, we provide concluding remarks and extensions in Section \ref{section:discussion}.

\section{Product Centered Dirichlet Processes} \label{section:PCDP}

\subsection{Multiview Clustering}
We assume that the data are split into two views $\X = (\X_1, \X_2)$, where $\X_v = (X_{v1}, \dots, X_{vn})$ and $X_{vi} \in \Omega_v$ for $v=1,2$, though our methodology can be generalized to more than two views. The views may be of varying dimension and data type. For instance, the views could be different omics measurements or an outcome and covariates. The dimension of each view is denoted by $d_v$ for all $v=1, \dots, V$, and the variables in the $v$th view are indicated with the index $j_v$. Associated with the $ith$ observation is parameter $\theta_{i} = (\theta_{1i}, \theta_{2i}) \in \Theta = \Theta_1 \times \Theta_2$ for view-specific parameter spaces $\Theta_1, \Theta_2$. We model the data using a multiview mixture model in which $\theta_i$ determines the probability distribution of $X_i = (X_{1i}, X_{2i})$. That is,
\begin{align} \label{eq:MMM}
    X_i \mid \theta_i & \sim f(X_i ; \theta_i), & \theta_i \mid P & \sim P, & \tx{for } i&=1, \dots, n; 
\end{align}
where $f(\cdot ; \theta)$ is a parametric probability distribution that depends on $\theta$ and $P$ is a probability distribution with support on $\Theta$, often called the mixing measure. We will assume that $P$ is almost surely discrete, so that draws $\theta_1, \dots, \theta_n \sim P$ are guaranteed to have unique values $\theta_1^*, \dots, \theta_K^*$ for some $K<\infty$. Additionally, there will be ties within each view represented by the unique values $\theta_{v1}^*, \dots, \theta_{vK_v}^*$.

We create \textit{clusters} by grouping observations with identical atoms. Let $C_{vk_v} = \lb i: \theta_{vi} = \theta_{vk_v}^* \rb$ be the $k_v$th cluster in the $v$th view. The view-specific partitions are defined as $C_v = (C_{v1}, \dots, C_{vK_v})$ with associated labelings similarly defined as $\bs c_v = (c_{v1}, \dots, c_{v_n})$, where $c_{vi} = k_v \iff i \in C_{vk_v}$. The size of $C_{vk_v}$ is denoted by $n_{vk_v} = |C_{vk_v}|$, which we collect into the vector $\n_v = (n_{v1}, \dots, n_{vK_v})$. While we consider ties in $\theta_{i1}$ and $\theta_{i2}$ separately, we will also observe ties in $\theta_i$, i.e. $\theta_i = \theta_j$ $\iff$ $\theta_{i1} = \theta_{j1}$ and $\theta_{i2} = \theta_{j2}$.  These ties are represented by a \textit{cross-partition} or \textit{cross-clustering} $D = (D_{k_1k_2})_{k_1, k_2}$, where $D_{k_1k_2} = C_{1k_1} \cap C_{2k_2}$ for $k_1 = 1, \dots, K_1$ and $k_2 = 1, \dots, K_2$. The cross-partition is an object describing the joint distribution of two random partitions, analogous to collecting random variables into a random vector. The cross-clustering is related to the \textit{contingency table} between $C_1$ and $C_2$, given by $\n = (n_{k_1k_2})_{k_1, k_2}$ where $n_{k_1k_2} = |D_{k_1k_2}|$. The contingency table can be reverted to the cluster sizes using summations, i.e. $n_{1k_1} = \sum_{k_2} n_{k_1 k_2}$ and $n_{2k_2} = \sum_{k_1} n_{k_1 k_2}$. As we will discuss later on, $\n$ will be crucial for interpretation and computation of our model.

Observe that we can rewrite the distribution of $X_i$ in \eqref{eq:MMM} as $(X_i \mid c_{1i}=k_1, c_{2i} = k_2, \theta_{1k_1}^*, \theta_{2k_2}^*) \sim f(X_i; \theta_{1k_1}^*, \theta_{2k_2}^*)$. 
We address two cases in which dependent clustering may arise. The first is when the views are assumed to be independent conditional on the clusterings,
\begin{equation} \label{eq:multiview-model}
    X_i \mid c_{1i}=k_1, c_{2i} = k_2, \theta_{1k_1}^*, \theta_{2k_2}^* \sim f_1(X_{1i}; \theta_{1k_1}^*) \times f_2(X_{2i}; \theta_{2k_2}^*),
\end{equation}
where $f_v(\cdot ; \theta_v)$ is a parametric probability distribution for the $v$th view. Here, we can interpret $C_1$ as a clustering for $\X_1$ and $C_2$ as a clustering for $\X_2$. In the case of multiomics, $C_1$ may be a partition related to SNPs, whereas $C_2$ represents a partition induced by RNA expression. The kernel in \eqref{eq:multiview-model} is equivalent to assuming that the relationship between $C_1$ and $C_2$ comprises the overall across view dependence structure. If, for example, $C_1$ and $C_2$ are independent, then $X_{1i} \ind X_{2i}$.

Conversely, for covariate-outcome data, it may be more appropriate to model the data with
\begin{equation} \label{eq:conditional-model}
    X_i \mid c_{1i}=k_1, c_{2i} = k_2, \theta_{1k_1}^*, \theta_{2k_2}^* \sim f_1(X_{1i}; \theta_{1k_1}^*) \times f_2(X_{2i}; X_{1i}, \theta_{2k_2}^*),
\end{equation}
where $f_2(\cdot; \cdot, \theta_2)$ is a parametric conditional distribution and $f_1(\cdot; \theta_1)$ is the marginal distribution for $\X_1$. We then interpret $C_1$ as $\X_1$ clusters, but $C_2$ are now $(\X_2 \mid \X_1)$ clusters, i.e. clusters related to the conditional distribution of the second view given the first. For example, if $X_{i1}, X_{2i} \in \mathbb{R}$, \eqref{eq:conditional-model} can be modeled with a linear regression model $f_2(X_{2i}; X_{1i}, \theta_{2k_2}^*) = \N(X_{2i}; X_{1i} \theta_{2k_2}^* , \sigma^2)$. $C_2$ groups together observations with a similar relationship between the covariates $\X_1$ and the response $\X_2$. If $C_1 \ind C_2$, then groups based off the covariates are not useful for prediction. Observe that if there are $V$ sub-views amongst the covariates $\X_1$, we can model $f_1(X_{1i}; \theta_{1k_1}^*)$ with the product kernel in \eqref{eq:multiview-model}, resulting in $V+1$ dependent partitions comprising the outcome and the covariate sub-views. 

The main motivation of this article is to address the dependence structure between $C_1$ and $C_2$. On one end of the spectrum, $C_1 \ind C_2$, meaning that the clustering of one view gives no information about the clustering in the other. There are a wide variety of methods to model independent clusterings for \eqref{eq:multiview-model} and \eqref{eq:conditional-model} using random measures, such as setting $P=P_1 \times P_2$, where $P_1 \ind P_2$ and each $P_v$ follows a DP prior. On the other hand, the two clusterings could be trivially dependent, i.e. $C_1=C_2$ almost surely. Trivially dependent clusterings imply that the underlying grouping structure of $\X$ is unchanged between the views. If $C_1 = C_2$, one can see that $\n$ is a generalized permutation matrix, i.e. there is only one non-zero entry in each row and column of the contingency table. This clustering model can be implemented by assuming $P$ is drawn from a DP prior with a non-atomic base measure, as ties in $\theta_{1i}$ are equivalent to ties in $\theta_{2i}$. We are interested in defining a model bridging between independent and identical clusterings across the views, allowing nontrivial dependence between $C_1$ and $C_2$.

\subsection{Prior on the Mixing Measure}
The product centered Dirichlet process (PCDP) is a nonparametric prior for random measures that can model dependence between partitions. To begin, we assume that the parameter space is a Cartesian product of view-specific parameter spaces, i.e. $\Theta = \Theta_1 \times \Theta_2 = \lb (\theta_{1k_1}^*, \theta_{2k_2}^*) : \theta_{vk_v}^* \in \Theta_v \rb$. At the top level, we draw $P$ from a DP conditional on a concentration parameter $\rho$ and a base measure $Q$,
\begin{equation} \label{eq:P-prior}
    P \mid \rho, Q \sim \DP(\rho, Q).
\end{equation}
Since $P$ is drawn from a DP, it is guaranteed to be almost surely discrete \citep{sethuraman1994constructive}, which ultimately leads to clusters in the data generating process. This procedure can be better visualized with the P\'olya urn scheme \citep{blackwell1973ferguson}, which computes the conditional distribution of the $i$th atom after marginalizing out $P$:
\begin{equation} \label{eq:P-polya}
    \theta_i \mid \theta_1, \dots, \theta_{i-1}, \rho, Q \sim \sum_{j=1}^{i-1} \left( \frac{1}{\rho + i -1} \right) \delta_{\theta_j} + \frac{\rho}{\rho + i - 1} Q ,
\end{equation}
where $\delta_{\theta_j}$ is the degenerate measure at $\theta_j$. With probability $1/(\rho + i - 1)$, we set $\theta_i$ equal to one of the already generated atoms $\theta_j$, placing $i$ into an already created cluster. Otherwise, with probability $\rho/(\rho + i - 1)$, we simulate $\theta_i \sim Q$, generating a potentially new value.

Observe that $Q$ must be carefully chosen in order to induce non-trivially dependent clusters. Suppose that $Q$ is some non-atomic distribution, i.e. $Q(\{ (\theta^*_{1k_1}, \theta^*_{1k_2}) \}) = 0$ for any pair of atoms. Under this regime, \eqref{eq:P-polya} implies that $C_1 = C_2$ almost surely since $\theta_{1i} = \theta_{1j} \iff \theta_{2i} = \theta_{2j}$. When dealing with the standard DP and non-atomic base measure, there are only two choices when sampling the $ith$ atom in \eqref{eq:P-polya}: sample $\theta_i$ from an existing cluster or create a new cluster. This scheme is visualized using the Chinese restaurant process (CRP) metaphor. Here, $\theta_i$ are customers entering a Chinese restaurant with an infinite amount of tables. They can either sit at one of the occupied tables, indicated by $\theta_i = \theta_j$, or sit at an unoccupied table, which corresponds to sampling $\theta_i \sim Q$ \citep{aldous1985exchangeability}. In contrast, the cross-clustering case motivates four choices: set $\theta_{1i}$ and $\theta_{2i}$ equal to existing atomic values, set $\theta_{1i} = \theta_{1j}$ and generate a new value of $\theta_{2i}$, set $\theta_{2i} = \theta_{2j}$ and generate a new value of $\theta_{1i}$, or generate two new values of $\theta_{1i}$ and $\theta_{2i}$. In a dependent clustering model, each of these possibilities should have positive probability when simulating $\theta_i$. Alternatively, if we were to simulate $\theta_{1i} \sim P_1$ and $\theta_{2i} \sim P_2$, where $P_v \sim \DP(\rho_v, Q_v)$ independently, then there would indeed be the desired four choices when sampling $\theta_i$, but $C_1 \ind C_2$.

A middle ground between the independent and trivially dependent case is induced by assigning $Q$ an appropriate prior distribution. The first step is to construct $Q$ using marginals on $\Theta_v$. Let $Q_v$ be probability distributions on $(\Theta_v, \mathcal{S}_v)$, where $\mathcal{S}_v$ is a $\sigma$-algebra on $\Theta_v$, for $v=1,2$. The product measure $Q = Q_1 \times Q_2$ is defined by $Q(S_1 \times S_2) = Q_1(S_1)Q_2(S_2)$ for any $S_1 \in \mathcal{S}_1$ and $S_2 \in \mathcal{S}_2$. Occasionally, the notation $Q=\bigtimes_{v=1}^2 Q_v$ may be used to refer to the product measure. We will furthermore assume that $Q_v$ are also drawn from DP priors, i.e.
\begin{equation} \label{eq:Q-prior}
    Q_v \mid \gamma_v, H_v \sim \DP(\gamma_v, H_v) \tx{ independently for } v=1,2.
\end{equation}
We will typically take $H_v$ to be non-atomic. Since $Q_v$ are themselves probability measures, $Q$ is guaranteed to exist and be unique almost surely \citep{resnick2001probability}, and its marginal distributions are $Q_1$ and $Q_2$. Observe that if $\rho \to \infty$, then $P=Q$ and, therefore, $C_1 \ind C_2$. For this reason, we interpret $Q$ as the null case, with degree of concentration controlled by $\rho$, and call \eqref{eq:P-prior} and \eqref{eq:Q-prior} the PCDP for modeling random measures. Observe that there are three concentration parameters in this process: $\gamma_1$, $\gamma_2$, and $\rho$. In the Supplemental Material, we verify that these hyperparameters are identifiable for $\pi (P, Q_1, Q_2 \mid \rho, \gamma_1, \gamma_2)$ for any non-atomic base measures $H_1$ and $H_2$. In practice, we fix $\gamma_1$ and $\gamma_2$, but a similar argument implies identifiability for $\rho$ in this case.

This formulation is distinct from the HDP \citep{teh2006sharing}, which assumes that $P$ follows \eqref{eq:P-prior} and $Q \sim \DP(\gamma, H)$, with no specification on the marginals. One can show that when $Q_1$ and $Q_2$ follow \eqref{eq:Q-prior}, their product measure is not distributed as a DP. For example, let $S_v \in \mathcal{S}_v$ for $v=1,2$. Then we construct a partition of $\Theta$ using the sets $S_1 \times S_2$, $S_1 \times S_2^c$, $S_1^c \times S_2$, and $S_1^c \times S_2^c$, where $S_v^c = \Theta_v \setminus S_v$. We then have, for instance, that $Q(S_1 \times S_2) = Q_1(S_1)Q_2(S_2)$, where $Q_v(S_v) \sim \tx{Beta}(\gamma_v H_v(S_v), \gamma_v H_v(S_v^c))$ for $v=1,2$ and are independent. Hence, $Q(S_1 \times S_2)$ is not distributed as a Beta random variable and instead follows a Wilks' type-B distribution \citep{wilks1932certain, tang1984distribution}. On the other hand, the marginals of $P$ will be distributed as HDPs. Define $P_1(\cdot) = \int_{\Theta_2} \diff P(\cdot, \theta_2)$ and $P_2(\cdot) = \int_{\Theta_1} \diff P(\theta_1, \cdot)$, to be the marginal distributions of $P$ on $\Theta_v$ for $v=1,2$. For example, conditional on $Q_1,Q_2$ and given the partition $S_{11}, \dots, S_{1m}$ of $\Theta_1$, Fubini's theorem implies
\begin{gather} 
    (P_1(S_{11}), \dots, P_{1}(S_{1m}))  \overset{d}{=} (P(S_{11} \times \Theta_2), \dots, P(S_{1m} \times \Theta_2)) \nonumber \\
    \sim \tx{Dir}(\rho Q_1(S_{11}), \dots, \rho Q_1(S_{1m})), \label{eq:P-marginal}
\end{gather}
and a similar property holds for $P_2$. In other words, the prior for the marginals is $P_v \mid \rho, Q_v \sim \DP(\rho, Q_v)$, $Q_v \sim \DP(\gamma_v, H_v)$. 

Using \eqref{eq:P-polya} and \eqref{eq:Q-prior}, we can express a P\'olya urn scheme for the PCDP. For observation $i$, we  observe $K^{(i)} \leq i-1$ unique values in $\theta^{(i)} = (\theta_1, \dots, \theta_{i-1})$ since $P$ is almost surely discrete. Across the views, this translates into unique atoms $\theta_{v1}^*, \dots, \theta_{vK_v^{(i)}}^*$, where $K_v^{(i)} \leq K^{(i)}$ for $v=1,2$. In order to compute the predictive distribution of $\theta_i$, we must marginalize out $Q$ from \eqref{eq:P-polya}. Since $Q$ is almost surely discrete and a product measure, we derive a unique form of the usual P\'olya urn scheme for Bayesian clustering.
\begin{theorem} \label{thm:thm:polya-urn}
     The distribution of the $i$th atom conditional on the first $i-1$ atoms is
    \begin{equation} \label{eq:polya}
    \theta_i \mid \theta^{(i)}, \rho \sim  \sum_{k_1, k_2}  \frac{n_{k_1 k_2}}{\rho + i - 1} \delta_{(\theta_{1k_1}^*, \theta_{2k_2}^*)} + \frac{\rho}{\rho + i - 1} \bigtimes_{v=1}^2 \lb  \sum_{k_v=1}^{K_v^{(i)}} \frac{r_{vk_v} \delta_{\theta_{vk_v}^*}}{\gamma_v + r^{(i)}}   + \frac{\gamma_v H_v}{\gamma_v + r^{(i)}}  \rb,
    \end{equation}
    where $\bs r_v = (r_{v1}, \dots, r_{vK_v^{(i)}})$, $r_{vk_v} \in \lb 0, \dots, n_{vk_v} \rb$, and $\sum_{v=1}^{K_v^{(i)}} r_{vK_v} = r^{(i)}$ for $v=1,2$.
\end{theorem}
The two terms in \eqref{eq:polya} can be interpreted in terms of an independence centered CRP metaphor. Suppose that $\theta_{1i}$ denotes the table and $\theta_{2i}$ denotes the entr\'ee of the $i$th customer. With probability $1/(\rho + i - 1)$, the customer sits at an occupied table and, seeing an appetizing dish already placed at the table, orders the same entr\'ee as one of their tablemates, which is indicated by setting $\theta_i = \theta_j$. The probability of each option is weighted by their corresponding entry in the contingency table, so popular combinations of table and entr\'ee will have more of an influence on customer $i$'s choice. Otherwise, with probability $\rho/(\rho + i - 1)$, the customer chooses a table and orders an entr\'ee independently. By the time customer $i$ has entered the restaurant, $r^{(i)} \leq i-1$ customers have made their table and entr\'ee choices independently. Of these $r^{(i)}$ people, $r_{1k_1}$ have sat down at table $k_1$ and $r_{2k_2}$ have ordered entr\'ee $k_2$. The vector $\bs r_1$ is just the number of customers at each table who made their choices independently, and $\bs r_2$ is the same but across the entr\'ees. If customer $i$'s choice of table and entr\'ee are independent, then the probability of sitting at table $k_1$ is weighted by $r_{1k_1}$, and a similar property holds for the entr\'ee with $r_{2k_2}$.  

We can then update \eqref{eq:polya} for customer $i+1$ by adhering to the following rules. If customer $i$ selects table $k_1$ and entr\'ee $k_2$ from the contingency table $\bs n$, we set $n_{k_1k_2} = n_{k_1k_2} + 1$. Otherwise, customer $i$ chooses independently and we set $r^{(i+1)} = r^{(i)} + 1$. If customer $i$ happens to choose table $k_1$ and entre\'e $k_2$, we set $n_{k_1k_2} = n_{k_1k_2}+1$, $r_{1k_1} = r_{1k_1}+1$, and $r_{2k_2} = r_{2k_2}+1$. Suppose now that, given customer $i$ chooses independently, they opt for table $k_1$ but order a new entre\'e. Then we would create a new entry $n_{k_1 (K_2^{(i)}+1)}=1$ in the contingency table, create a new entry of $\bs r_2$ given by $r_{2(K_2^{(i)}+1)}=1$, and set $r_{1k_1} = r_{1k_1}+1$. A similar rule holds for sitting at a new table but asking for an already ordered entr\'ee. Finally, if customer $i$ opts for both a new table and new entr\'ee, we create $n_{(K_1^{(i)}+1)(K_2^{(i)}+1)}=1$, $r_{1(K_1^{(i)}+1)}=1$, and $r_{2(K_2^{(i)}+1)}=1$.   

Clearly, the larger $\rho$ is, the less likely it is that the choice of table depends on the choice of entr\'ee, meaning that all dishes should be distributed independently across the restaurant, $r^{(i)} \to i-1$, and $\bs r_v \to \bs n_v$ for $v=1,2$. However, the marginal concentration parameters $\gamma_v$ will have an impact on these choices as well. If, instead, we took $\gamma_1,\gamma_2 \to \infty$, we will find the opposite: all the customers at each table will order the same entr\'ee, and no two tables will have the same dish. That is, the choice of food is equivalent to the choice of a table. This behavior can be seen by taking limits in \eqref{eq:polya}, but we can also state these limiting properties in terms of $P$.
\begin{prop} \label{prop:convergence}
Recall that $\mathcal{S}_v$ is a $\sigma$-algebra on $\Theta_v$ for $v=1,2$.
    \begin{enumerate}
        \item Fix $\gamma_1, \gamma_2$ and let $S \subset \Theta$. Then as $\rho \to \infty$, $P(S) \overset{p}{\longrightarrow} Q(S)$. 
        \item  Fix $\rho$ and let $S = S_1 \times S_2$, where $S_v \in \mathcal{S}_v$ and $H_v(S_v)>0$ for $v=1,2$. Then as $\gamma_1, \gamma_2 \to \infty$, $P(S) \overset{d}{\longrightarrow} P_0(S)$ where $P_0 \sim \DP(\rho, H_1 \times H_2)$.
    \end{enumerate}
\end{prop}

Proposition \ref{prop:convergence} states that in limiting cases for the concentration parameters, the PCDP converges to either a DP with a product base measure or the product measure of two independent DP measures. The first statement follows after applying Markov's inequality to show that for any $\epsilon>0$ and $\rho>0$, $\pi (|P(S)-Q(S)| > \epsilon \mid \rho) \leq 1/ \lb 4 \epsilon^2(\rho+1)\rb$. We verify the second statement by proving convergence of the CDF for $P(S)$ conditional on $Q_1(S_1)$ and $Q_2(S_2)$, then applying the dominated convergence theorem. Therefore, the choices of $\rho$, $\gamma_1$, and $\gamma_2$ are very important for the distribution of the random partitions. We will return to the effect of the concentration parameters on inference in Section \ref{section:CLIC}, where we show that in the PCDP, the dependence between $C_1$ and $C_2$ is modeled as a function of these values. 

The PCDP can be extended to $V \geq 2$ views by setting $Q=\bigtimes_{v=1}^V Q_v$, where $Q_v \mid \gamma_v, H_v \sim \DP(\gamma_v, H_v)$, then sampling $P \mid \rho, Q$ from \eqref{eq:P-prior}. Here, $\rho$ can be interpreted as a global clustering dependence parameter, i.e., the dependence structure between each pair $(C_v, C_u)$ is identical for all $v \neq u$. Additionally, a similar P\'olya urn scheme to \eqref{eq:polya} holds in the general case, where past combinations are weighted by entries in the contingency array $\n = (n_{k_1 \cdots k_V})_{k_1, \dots, k_V}$ and the customer now has $V$ choices to make when they enter the restaurant. Accordingly, the $V$-view cross-partition is defined as $D = (D_{k_1 \dots k_V})_{v=1, \dots, V}$, where $D_{k_1 \dots k_V} = \bigcap_{v=1}^V C_{vk_v}$. We can recover the contingency tables for any two views by summing along the dimensions of $\n$, and we can construct the associated $2$-view cross-partitions by taking unions of the sets in $D$. By Fubini's theorem, the univariate marginals of $P$ follow an HDP, and pairwise marginals are distributed as PCDPs with two views. Therefore, results from the two view case will carry over to any pairs of views in the $V$-partition setting.

\subsection{Finite Approximation}
Draws from the PCDP are random measures with a countably infinite number of support points, i.e. $P = \sum_{k_1=1}^\infty \sum_{k_2=1}^\infty p_{k_1 k_2} \delta_{(\theta_{1k_1}^*, \theta_{2k_2}^*)}$. However, this makes simulation and inference on $P$ computationally infeasible. In this section, we propose a finite measure $\tilde{P} = \sum_{k_1=1}^{L_1} \sum_{k_2=1}^{L_2} \tilde{p}_{k_1 k_2} \delta_{(\tilde{\theta}^*_{1k_1}, \tilde{\theta}^*_{2k_2})}$, where $L_1, L_2 < \infty$, which approximates $P$ as $L_1$ and $L_2$ grow large. Our method is similar to the approximation to the DP prior given by sampling a $L$-dimensional $p \sim \tx{Dir}(\rho/L, \dots, \rho/L)$ and $\theta_l^* \sim Q$ for $l=1, \dots, L$, then taking $L \to \infty$ \citep{ishwaran2002exact, lijoi2023finite}. \cite{teh2006sharing} showed several finite approximations to the HDP using the Dirichlet-multinomial distribution, which we adapt for the PCDP. The finite PCDP is given by
\begin{equation}
    \begin{aligned} \label{eq:finite-approximation}
    \tilde{p} \mid \rho, \tilde{q} & \sim \tx{Dir}(\rho \tilde q); & \tilde{q} & = \tilde{q}_1 \otimes \tilde{q}_2; \\
    \tilde{q}_v & \sim \tx{Dir}(\gamma_v/L_v, \dots, \gamma_v/L_v), & \tilde{\theta}^*_{vk_v} & \sim H_v \tx{ for } v=1,2;
\end{aligned}
\end{equation}
where $\otimes$ denotes the tensor product, i.e. $\tilde{q} = (\tilde{q}_{1k_1} \tilde{q}_{2k_2})_{k_1,k_2}$, and $k_v = 1, \dots, L_v$. We denote $\tilde{Q}_v = \sum_{k_v=1}^{L_v} \tilde q_{vk_v} \delta_{\tilde{\theta}_{vk_v}^*}$ to be the approximation of $\tilde{Q}_v$ and set $\tilde{Q} = \tilde{Q}_1 \times \tilde{Q}_2$. Using Fubini's theorem and results on DP approximations \citep{ishwaran2002exact}, we obtain the following result on 
convergence of $\tilde{Q}$ to $Q$. Our result is stated in terms of the expectation of a bounded measurable function $h(\theta)$. We denote $\int \int h(\theta) \diff \tilde{Q}(\theta) = \E_{\tilde Q}[h(\theta)]$ and $\int \int h(\theta) \diff Q(\theta) = \E_{Q}[h(\theta)]$.
\begin{prop} \label{prop:finite-approximation}
Let $h:\Theta \to \Theta^\prime \subset \mathbb{R}$ so that $h(\theta) = h_1(\theta_1)h_2(\theta_2)$, $h_v(\cdot)$ is $(\Theta_v, \mathcal{S}_v)$-measurable for $v=1,2$, and $h$ is bounded on $\Theta$. Then, as $L_1,L_2 \to \infty$,
\begin{equation} \label{eq:double-convergence}
    \int \int h(\theta) \diff \tilde{Q}(\theta) \overset{d}{\longrightarrow} \int \int h(\theta) \diff Q(\theta).
\end{equation}
\end{prop}

With regard to $\tilde{P}$, observe that, as in the infinite case, $\tilde{P} \mid \rho, \tilde{Q} \sim \DP(\rho, \tilde{Q})$: if $S_1, \dots, S_K$ is a partition of $\Theta$, then $(\tilde{P}(S_1), \dots, \tilde{P}(S_M)) \mid \rho, \tilde{Q} \sim \tx{Dir}(\rho \tilde{Q}(S_1), \dots, \rho \tilde{Q}(S_M))$, where $\tilde{Q}(S_m) = \sum_{k_1, k_2} \textbf{1}_{(\tilde{\theta}^*_{1k_1}, \tilde{\theta}^*_{2k_2}) \in S_m } \tilde{q}_{1k_1} \tilde{q}_{2k_2}$. 
Examples of functions $h$ satisfying the conditions of Proposition \ref{prop:finite-approximation} are the CDFs of $X_i$ in the correlated and uncorrelated views models, i.e. \eqref{eq:multiview-model} and \eqref{eq:conditional-model}. This shows that the distribution of $X_i$ under the finite model will resemble the infinite case as $L_1, L_2 \to \infty$. For example, the expected CDF of $X_i$ in the finite approximation is $\tilde{\E}[\tilde{\Pi}(x_i \mid \tilde{P}) \mid \tilde Q]= \int \int F(x_i; \theta ) \diff \tilde{Q}(\theta)$ and in the regular PCDP is $\E[\Pi(x_i \mid P) \mid Q] = \int \int F(x_i; \theta ) \diff Q(\theta)$, where $F(x_i;\theta)$ is the CDF for $f(x_i;\theta)$. Then, \eqref{eq:double-convergence} implies that $\tilde{\E}[\tilde{\Pi}(x_i \mid \tilde{P}) \mid \tilde Q] \overset{d}{\to} \E[\Pi(x_i \mid P) \mid Q]$.

Properties from the infinite PCDP carry over into the finite approximation. For example, the marginals $\tilde{P}_v$ of $\tilde{P}$ are distributed as finite approximations to the HDP, and using a similar argument to the one above, we can show that these provide accurate approximations to the view-specific marginal CDFs. In Section \ref{section:computation}, we derive a Gibbs sampler for multiview and conditional models using the finite approximation in \eqref{eq:finite-approximation}. In practice, we recommend setting $L_1$ and $L_2$ to be large positive integers that serve as prior beliefs on an upper bound to the number of clusters $K_1$ and $K_2$ along the views. We will show in the following section that the expectation of the Rand index for the finite approximation converges to that of the PCDP at a $(L_1L_2)^{-1}$ rate.

\section{Clustering with Independence Centring} \label{section:CLIC}
Recall that ties in $\theta_{vi}$ generate a partition $C_v = (C_{v1}, \dots, C_{vK_v}) \in \Pn$ for $v=1,2$. In this section, we address the induced joint distribution of $(C_1,C_2)$, which we denote as $\CLIC(\rho, \gamma_1, \gamma_2)$. We will show that the joint distribution of $(C_1,C_2)$ depends only on the contingency table $\n = (n_{k_1 k_2})_{k_1,k_2}$, where $n_{k_1k_2} = |C_{1k_1} \cap C_{2k_2}|$. This behavior is comparable to an exchangeable partition probability function (EPPF), a special type of PMF for random partitions that is an exchangeable function of the cluster sizes \citep{pitman2002combinatorial}. The cluster sizes can be recovered via summation, i.e. $n_{1k_1} = \sum_{k_2}n_{k_1k_2}$ (row summation) and $n_{2k_2} = \sum_{k_1} n_{k_1k_2}$ (column summation). The $\CRP(\gamma)$ distribution in \eqref{eq:CRP-PMF} is an example of an EPPF, though partitions induced by Pitman-Yor processes, symmetric Dirichlet-Multinomial sampling, and even hierarchical processes such as the HDP also admit EPPFs \citep{camerlenghi2018bayesian,camerlenghi2019distribution}. The telescopic EPPF \citep{franzolini2023bayesian} is a generalization of EPPFs to multiview clustering, and closed-form expression of telescopic EPPFs exist for several multiview clustering models. 

We say that the joint partition function of $(C_1,C_2)$ is a marginally exchangeable partition probability function (MEPPF) if $\pi(C_1,C_2)$ is a function of the contingency table $\n$ that is invariant to permutations of the \textit{marginal} cluster indices. For example, permutations of the rows or columns of the contingency table will result in the same prior probability, i.e. if $\pi(C_1,C_2) = g( \n)$ and if $\omega$ is a permutation of $[K_1]$ and $\n^\omega = (n_{\omega(k_1), k_2})_{k_1,k_2}$, then $g(\n) = g(\n^\omega)$; an identical guarantee holds for permutations of $[K_2]$. Observe that exchangeability of the marginal cluster indices is not equivalent to exchangeability of the entries in $\n$. Pragmatically, this is to ensure that we can always recover the marginal cluster sizes by summing along the rows or columns of $\n$. We now present the closed-form expression for the MEPPF of the $\CLIC(\rho, \gamma_1, \gamma_2)$ distribution. 
\begin{theorem} \label{thm:joint-EPPF}
    Suppose $\mathcal{S}_v$ is the Borel algebra of $\Theta_v$ for $v=1,2$. If $(C_1,C_2) \sim \CLIC(\rho, \gamma_1, \gamma_2)$, then
    \begin{gather} \label{eq:joint-EPPF}
        \pi(C_1,C_2 \mid \rho, \gamma_1, \gamma_2) = \frac{\gamma_1^{K_1}\gamma_2^{K_2}}{(\rho)^{(n)}} \sum_{ \bs r \in \mathcal R} \rho^{r} \lb \prod_{v=1}^{2}  \frac{\prod_{k_v} (r_{vk_v}-1)!}{(\gamma_v)^{(r)}} \rb \prod_{k_1,k_2} |s(n_{k_1 k_2}, r_{k_1 k_2})|,
    \end{gather}
    where $(\rho)^{(n)} = \Gamma(\rho + n)/\Gamma(\rho)$, $\mathcal R =  \lb \bs r =  (r_{k_1 k_2})_{k_1,k_2}: r_{k_1k_2} \in [n_{k_1, k_2}] \: \forall k_1,k_2 \rb$, $r = \sum_{k_1,k_2} r_{k_1 k_2}$, $r_{1k_1} = \sum_{k_2}r_{k_1k_2}$, $r_{2k_2} = \sum_{k_1} r_{k_1k_2}$, and $|s(n,m)|$ are the unsigned Stirling numbers of the first kind.
\end{theorem}

We present the proof of Theorem \ref{thm:joint-EPPF} in the Supplemental Material; the result follows by adapting the derivations of \cite{camerlenghi2018bayesian} on the EPPF induced by the HDP and exploiting the product measure formulation of $Q$ and the prior independence of $Q_1$ and $Q_2$. Observe that $\pi(C_1, C_2 \mid \rho, \gamma_1, \gamma_2)$ is itself a mixture model over independent contingency tables in $\mathcal{R}$, given by the \textit{root partitions} $T_1$ and $T_2$. First, we sample $r \in [n]$ via $\pi(r=w \mid \rho) \propto \rho^w$ for $w =1, \dots, n$, then we sample two partitions of $[r]$ from
\begin{equation} \label{eq:root-partitions}
    \pi(T_1, T_2 \mid \gamma_1, \gamma_2) = \prod_{v=1}^2 \frac{\gamma_v^{K_v} \prod_{k_v} (r_{vk_v}-1)!}{(\gamma_v)^{(r)}},
\end{equation}
which is the product of the EPPFs for two independent $\CRP(\gamma_v)$ partitions. Finally, we sample $(C_1, C_2)$ conditionally on $(T_1,T_2)$ from
\begin{equation} \label{eq:C1-C2-conditional}
    \pi(C_1, C_2 \mid \rho, \gamma_1, \gamma_2, T_1, T_2) \propto \frac{\Gamma(\rho)}{\Gamma(n + \rho)} \prod_{k_1,k_2} |s(n_{k_1 k_2}, r_{k_1 k_2})|.
\end{equation}
Hence, we can reinterpret the data generating process as first drawing two root partitions, then perturbing their cross partition using the probability mass function in \eqref{eq:C1-C2-conditional}. An important consequence of this decomposition is that
\begin{equation} \label{eq:Cv-number-of-clusters}
    \pi(K_v = m \mid \rho, \gamma_v) = \begin{cases}
        \sum_{w=1}^n \frac{\rho^{w-1}(1-\rho)}{1-\rho^n} \Pr[ K(w; \gamma_v) = m]  & \tx{ for } \rho \neq 1; \\
       \frac{1}{n} \sum_{w=1}^n \Pr[ K(w; \gamma_v) = m] & \tx{ for } \rho = 1;
    \end{cases}, 
\end{equation}
where $K(w; \gamma_v)$ is the number of clusters in a sample of $w$ objects from a $\CRP(\gamma_v)$ distribution, see \cite{antoniak1974mixtures} and \cite{gnedin2006exchangeable} for further information on this quantity. Therefore, both $\rho$ and $\gamma_v$ impact the number of clusters in the marginals. As $\rho \to \infty$, $\rho^{n-1}(1 - \rho)/(1-\rho^n) \to 1$, meaning that $K_v \overset{d}{\longrightarrow} K(n; \gamma_v)$, as expected by our discussion in the previous section. Finally, $K_v$ grows very slowly with the sample size a priori, as it has been shown that the number of clusters in an HDP scales asymptotically at a $\log \log n$ rate \citep{teh2010hierarchical,camerlenghi2019distribution}.

The marginal distributions of $P$ in \eqref{eq:P-marginal} imply that $C_1$ and $C_2$ are both distributed as HDP partitions, and if $\gamma_1=\gamma_2$, then $C_1 \overset{d}{=} C_2$. The EPPF of $C_v$ is then given by
\begin{equation} \label{eq:marginal-EPPF}
    \pi(C_v \mid \rho, \gamma_v) = \frac{\gamma_v^{K_v}}{(\rho)^{(n)}} \sum_{r_v \in \mathcal R_v} \frac{\rho^{r}}{(\gamma_v)^{(r)}} \prod_{k_v=1}^{K_v}(r_{vk_v}-1)! |s(n_{vk_v}, r_{vk_v})|, \tx{ for } v=1,2;
\end{equation}
where the domain of the sum is $\mathcal R_v = \lb  \bs r_v = (r_{v1}, \dots, r_{vK_v}) : r_{vk_v} \in [n_{vk_v}], k_v=1,\dots,K_v \rb$ \citep{camerlenghi2018bayesian}. Similar to the joint distribution, the marginals admit a conditional sampling scheme given a root partition $T_v$. As we also showed in \eqref{eq:Cv-number-of-clusters}, $\rho$ affects not only the joint distribution of the partitions but also the marginals. This phenomenon is evident in the P\'olya urn scheme in \eqref{eq:polya}. The probability of choosing a table and dish independently, which includes the possibility of picking a new table or dish, is controlled by $\rho$. Hence, for very small $\rho$, the number of clusters in both $C_1$ and $C_2$ is expected to be small. 

Following \cite{page2022dependent} and \cite{franzolini2023conditional}, we use the prior expected Rand index (ERI) $\tau_{12} = \sum_{C_1, C_2 \in \Pn} R(C_1, C_2) \pi(C_1, C_2 \mid \rho, \gamma_1, \gamma_2)$ to measure dependence between $C_1$ and $C_2$. For convenience, we will assume that $\gamma_1 = \gamma_2 = \gamma$, but a similar result holds for differing values of $\gamma_v$. For two independent CRP partitions, their ERI is given by $\tau_{12}^0 = (1+\gamma^2)/(1+\gamma)^2$ \citep{page2022dependent}. Under the PCDP, the ERI is a weighted average between the ERI of two independent CRP partitions and $1$, the supremum of the Rand index.
\begin{theorem} \label{thm:ERI}
    Let $(C_1, C_2) \sim \CLIC(\rho, \gamma, \gamma)$. Then,
    \begin{equation} \label{eq:ERI}
        \tau_{12} = \frac{1}{\rho+1} + \left( 1 - \frac{1}{\rho+1} \right) \frac{1 + \gamma^2}{(1+\gamma)^2}.
    \end{equation}
\end{theorem}
Our intuition on the magnitude of the concentration parameters carries over to the ERI. As $\rho \to \infty, \tau_{12} \to (1+\gamma^2)/(1+\gamma)^2$, whereas if $\gamma \to \infty, \tau_{12} \to 1$. While the Rand index is guaranteed to be in $[0,1]$, we have that $(1+\gamma^2)/(1+\gamma)^2 \leq \tau_{12} \leq 1$. Similar truncation occurs in several models for clustering, which motivated the adjusted Rand index (ARI) \citep{hubert1985comparing} and, more recently, the telescopic adjusted Rand index (TARI) \citep{franzolini2023bayesian}, which removes the lower bound imposed by the independent model. In the latter case, we have that the expected TARI is simply equal to $\nu = 1/(\rho + 1)$ under the PCDP.

Recall that $\gamma$ controls the choices of customer $i$ in selecting a table and  
entr\'ee, with large $\gamma$ leading to trivial dependence between $C_1$ and $C_2$. We now have the added interpretation that $\gamma$ determines the ERI of the null independence model. That is, the ERI between $C_1$ and $C_2$ will be at least as large as the ERI between two independent $\CRP(\gamma)$ partitions. If we assign $\rho$ and $\gamma$ priors, we induce a prior directly on $\tau_{12}$. Since $\gamma$ is the DP concentration parameter of $Q_1,Q_2$ and $\rho$ is the DP concentration parameter of $P \mid Q$, we can model both using semi-conjugate prior distributions, such as a Gamma prior \citep{escobar1995bayesian} or a Stirling-Gamma prior \citep{zito2023bayesian}. Then, samples from the posterior distribution for the ERI $\pi(\tau_{12} \mid \X)$ can be obtained from the samples of $\rho,\gamma$.

A similar result on the ERI holds for the finite approximation to CLIC. In this setting, we sample $C_1,C_2$ using a matrix multinomial distribution, with $\tilde{\pi}(c_{1i} = k_1, c_{2i} = k_2 \mid \tilde{p}) = \tilde{p}_{k_1k_2}$ for all $k_1=1,\dots, L_1, \: k_2=1, \dots, L_2$. We show in the Supplement that if $\Tilde{\tau}_{12}$ is the ERI resulting from the model in \eqref{eq:finite-approximation}, then $\Tilde{\tau}_{12} = \nu + (1-\nu) \kappa(L_1,L_2)$, where $\kappa(L_1,L_2)$ is the ERI of two independent partitions with symmetric Dirichlet-multinomial priors and $\lim_{L_1,L_2 \to \infty} \kappa(L_1,L_2) = (1+\gamma^2)/(1+\gamma)^2$. This result also provides insight on the impact of setting $L_1$ and $L_2$ to finite values. For simplicity, suppose that $L_1 = L_2 = L$ for some fixed $L>0$. One can show that $| \tau_{12} - \tilde \tau_{12} | \propto | (\gamma/L) + 1 - \gamma  |/L$, where proportionality is in terms of $L$. That is, the rate of convergence for the ERI of the finite approximation to the PCDP is quadratic in $L^{-1}$. 

\section{Posterior Computation} \label{section:computation}

\subsection{Gibbs Sampler}

Given a dataset $\X$ and a multiview clustering model such as \eqref{eq:multiview-model} or \eqref{eq:conditional-model}, we focus on inferring the view-specific partitions $(C_1,C_2)$ and estimating the degree of dependence between the partitions using the Rand index or the TARI. Since the PCDP as presented in Section \ref{section:PCDP} has infinitely many mixture components, we present an algorithm specific to the finite approximation in \eqref{eq:finite-approximation}. We will assume that $H_v$ are conjugate to the kernels in either \eqref{eq:multiview-model} or \eqref{eq:conditional-model}, and that $\rho \sim \tx{Gamma}(a_\rho, b_\rho)$ for some $a_\rho, b_\rho > 0$ or $\pi(\rho = u) \propto 1$ for a set of finite numbers $u$ defined on a grid $U$, leading to a griddy-Gibbs sampler \citep{ritter1992facilitating}. The key strategy in our approach is to marginalize out the probability matrix $\tilde p$ from the multinomial likelihood of $(C_1,C_2)$, similar to posterior inference on the finite approximation to the HDP \citep{teh2006sharing}. This leads to a Gibbs sampler for the two view setting, but computation is similar when there are $V > 2$ views.

Let $K_v$ be the number of non-empty clusters in $C_v$, where $K_v \leq L_v$ for $v=1,2$. By marginalizing $\tilde p$ from $\pi(C_1,C_2 \mid \tilde{p}, \rho)$, one can show that the prior probability mass function of $(C_1,C_2)$ is given by
\begin{equation} \label{eq:S-Sprime-marginal}
    \pi(C_1,C_2 \mid \tilde q, \rho) = \frac{\Gamma(\rho)}{\Gamma(n + \rho)} \prod_{k_1=1}^{L_1} \prod_{k_2=1}^{L_2} \frac{\Gamma(n_{k_1k_2} + \rho \tilde q_{1k_1} \tilde q_{2k_2})}{ \Gamma(\rho \tilde q_{1k_1} \tilde q_{2k_2})  },
\end{equation}
where $n_{k_1 k_2} = 0$ if either $k_1 > K_1$ or $k_2 > K_2$. Denote $C_1^{-i}$ and $C_2^{-i}$ to be the partitions induced on $[n]\setminus \lb i \rb$ by removing observation $i$. We have that $\pi(c_{1i} = k_1, c_{2i} = k_2 \mid  C_1^{-i},  C_2^{-i}, \tilde q, \rho) \propto \rho \tilde q_{1k_1} \tilde q_{2k_2} + n_{k_1k_2}^{-i}$, where $n_{k_1 k_2}^{-i} = \sum_{j \neq i} \textbf{1}(c_{1j} = k_1, c_{2j} = k_2)$ for all $k_1 \in [K_1], k_2 \in [K_2]$ and $n_{k_1 k_2}^{-i} = 0$ if either $k_1 > K_1$ or $k_2 > K_2$. The Gibbs update of observation $i$ is given by the formula
\begin{equation}
    \pi(c_{1i} = k_1, c_{2i} = k_2 \mid -) \propto f(X_{i} ; (\theta_{1k_1}^*, \theta_{2k_2}^*))\left( \rho \tilde q_{1k_1} \tilde q_{2k_2} + n_{k_1k_2}^{-i}  \right),
\end{equation}
for all $k_1 \in [L_1], k_2 \in [L_2]$, where ``$-$" refers to $C_1^{-i},  C_2^{-i}, \tilde q, \rho, \theta^*_{1k_1}, \theta_{2k_2}^*$, and $\X$.

In order to sample the full conditional distribution of $\tilde q_1$, $\tilde q_2$, and $\rho$, we adapt an auxilary variable scheme used for inference on HDPs to our model \citep{teh2006sharing}. By using properties of the Gamma function, one can show that the product terms in \eqref{eq:S-Sprime-marginal} can be decomposed in the following way for all $k_1 \in [K_1], k_2 \in [K_2]$,
\begin{equation*}
    \frac{\Gamma(n_{k_1k_2} + \rho \tilde q_{1k_1} \tilde q_{2k_2})}{ \Gamma(\rho \tilde q_{1k_1} \tilde q_{2k_2})  }= \sum_{r_{k_1 k_2} = 0}^{n_{k_1 k_2}} |s(n_{k_1 k_2}, r_{k_1 k_2})| (\rho \tilde q_{1k_1} \tilde q_{2k_2})^{r_{k_1 k_2}},
\end{equation*}
where $|s(n,m)|$ are the unsigned Stirling numbers of the first kind. We can rewrite \eqref{eq:S-Sprime-marginal} as being the marginals of the following joint distribution based on a matrix $\bs r = (r_{k_1k_2})_{k_1,k_2}$ of auxiliary variables,
\begin{equation} \label{eq:auxiliary-Gibbs}
    \Pi(C_1, C_2, \bs r \mid \tilde q_1, \tilde q_2, \rho) = \frac{\Gamma(\rho)}{\Gamma(n + \rho)} \prod_{k_1=1}^{K_1} \prod_{k_2=1}^{K_2} |s(n_{k_1k_2}, r_{k_1 k_2})| (\rho \tilde  q_{1k_1} \tilde q_{2k_2})^{r_{k_1 k_2}}.
\end{equation}
Note that these auxiliary variables correspond to the domain of summation in the MEPPF \eqref{eq:joint-EPPF} for the PCDP, and that the product is over all non-empty clusters. Similarly, we define clustering-specific auxiliary variables $r_{1 k_1} = \sum_{k_2=1}^{K_2} r_{k_1 k_2}$ and $r_{2k_2} = \sum_{k_1=1}^{K_1} r_{k_1 k_2}$. To update $r_{k_1 k_2}$, we sample from a multinomial distribution using
\begin{align}
    \pi(r_{k_1 k_2} = w \mid C_1, C_2, \bs r^{- k_1 k_2}) & \propto |s(n_{k_1 k_2}, w)| (\rho \tilde q_{1k_1} \tilde q_{2k_2})^{w}; & w \in \lb 0, 1, \dots, n_{k_1 k_2} \rb.
\end{align}
We use the recursive formula $|s(n+1, w)| = n |s(n,w)| + |s(n,w-1)|$ to compute the unsigned Stirling numbers of the first kind.

Using \eqref{eq:auxiliary-Gibbs}, we can see that the full conditional distribution of $\tilde q_v$ is a $\Dir(\gamma_v/L_v + r_{v1}, \dots, \gamma_v/L_v + r_{vL_v})$, where $r_{vk_v} = 0$ for all $k_v > K_v$, for $v=1,2$. This update in particular provides insight onto the interpretation of the auxilary variables. Consider the independent clustering model, in which $\pi(c_{1i} = k_1, c_{2i} = k_2 \mid \tilde q_1, \tilde q_2) = \tilde q_{1k_1} \tilde q_{2 k_2}$. Then the Gibbs update of $\tilde q_1$ and $\tilde q_2$ will be exactly the same as in our model, except with $\bs r_1$ and $\bs r_2$ being replaced with $\bs n_1$ and $\bs n_2$. However, in general $r = \sum_{k_1,k_2} r_{k_1k_2} < n$. We then interpret $\bs r$ as being the cell counts of $C_1$ and $C_2$ if they were modeled as being truly independent and defined on a smaller sample. 

If $\rho \sim \tx{Gamma}(a_\rho, b_\rho)$, we use a data augmentation technique introduced in \cite{escobar1995bayesian}. From \eqref{eq:S-Sprime-marginal}, one sees that the full conditional distribution of $\rho$ is
\begin{equation} \label{eq:rho-marginal-fc}
    \Pi(\rho \mid C_1, C_2, \bs r, \tilde q_1, \tilde q_2, \X) \propto \pi(\rho) \frac{\Gamma(\rho)}{\Gamma(n+\rho)} \rho^{r} \propto \pi(\rho) \rho^{r-1}(\rho + n) \beta(\rho+1,n)
\end{equation}
where $\beta(\cdot, \cdot)$ is the beta function. Since $\beta(\rho+1,n) = \int_{0}^1 \eta^{\rho}(1-\eta)^{n-1}\diff \eta$, we have that \eqref{eq:rho-marginal-fc} is the marginal distribution of $\rho$ for a joint density proportional to $ \pi(\rho) \rho^{r-1}(\rho + n) \eta^\rho (1-\eta)^{n-1}$ for $0 < \eta < 1$ and $\rho > 0$. We set $\eta$ to be an auxiliary variable, then acquire samples from \eqref{eq:rho-marginal-fc} by alternating sampling of $\eta$ and $\rho$. We update $\eta$ by sampling from a $\tx{Beta}(\rho+1,n)$ distribution. To sample $\rho$, we simulate from a two-component mixture of Gamma distributions: $\omega_{\eta} \tx{Gamma}(a_\rho + r, b_\rho - \log \eta) + (1-\omega_\eta) \tx{Gamma}(a_\rho + r - 1, b_\rho - \log \eta),$ where $\omega_\eta/(1-\omega_\eta) = (a_\rho + r - 1) / \lb n(b_\rho - \log \eta) \rb$. If instead $\pi(\rho = u) \propto 1$ for all $u \in U$, we update the position of $\rho$ by simulating from a discrete distribution proportional to \eqref{eq:S-Sprime-marginal} for all points on the grid.

Finally, to update $\theta_{v k_v}^*$, we sample from $\pi(\theta_{v k_v}\mid \bs C_v, \X) \propto \lb \prod_{i \in C_{vk_v}} \pi(X_i ; \theta^*_{vk_v}) \rb H_v(\theta_{vk_v}^*)$ for $k_v = 1, \dots, K_v$, where in independent view case, $\pi(X_i; \theta_{vk_v}^*) = f_v(X_{vi}; \theta_{vk_v}^*)$, and in the correlated case, $\pi(X_i; \theta_{1k_1}^*) = f_1(X_{1i}; \theta_{1k_1}^*)$ and $\pi(X_i; \theta_{2k_2}^*) = f_2(X_{2i}; X_{1i}, \theta_{2k_2}^*)$. For any $k_v > K_v$, where $n_{vk_v}=0$, we simulate $\theta_{vk_v}^* \sim H_v$. 

\subsection{Point Estimation and Uncertainty Quantification}

At the termination of the Gibbs sampler we obtain MCMC samples $C_1^{(t)}, C_2^{(t)} \sim \pi(C_1, C_2 \mid \X)$ for $t=1, \dots, T$ iterations. We compute point estimates $\{ \widehat{C}_v \}$ of the partitions by minimizing the posterior expectation of a clustering loss function. In our simulations and application, we use the Variation of Information (VI) loss \citep{meilua2007comparing}, though alternatives exist including Binder's loss \citep{binder1978bayesian} and generalized losses \citep{dahl2022search}. As in the single view case, we would expect kernel misspecification in the PCDP to result in over-clustering the data by approximating true clusters with several mixture components. One strategy to robustify inference on the partitions is to apply a kernel merging algorithm to the marginals such as Fusing of Localized Densities (FOLD) \citep{dombowsky2023bayesian}, which takes the posterior samples of $\theta_{1i}$ and $\theta_{2i}$ as input. One could also implement the PCDP as a prior to infer the coarsened posterior of the partitions, an inferential framework which allows for perturbations in the observed data from the assumed component kernels \citep{miller2018robust}. 

Uncertainty in clustering is conveyed with either the posterior similarity matrix, which measures $\pi(c_{vi}=c_{vj} \mid \X)$ for all $i,j$ pairs, or a credible ball \citep{wade2018bayesian}, which provides a set of clusterings close to $\widehat C_v$ in the partition metric space that have high posterior probability. To quantify dependence in the partitions, we infer $\pi \lb R(C_1, C_2) \mid \X \rb$ or, alternately, the posterior of the TARI, and provide the posterior mean and a credible interval. For inference on the number of clusters in the $v$th view, we calculate the posterior probability of any $\hat K_v>0$ with $\pi(K_v = \hat K_v \mid \X) \approx (1/T)\sum_{t=1}^T \textbf{1}_{K_v^{(t)} = \hat K_v} $, where $K_v^{(t)}=\sum_{k_v=1}^{L_v} \textbf{1}_{C_{vk_v}^{(t)} \neq \emptyset}$ counts the number of non-empty blocks in the $v$th partition. For interpretability, we recommend inspecting the marginal posterior distributions of $K_1$ and $K_2$ individually, though inferences on their joint distribution can also be obtained with $\pi(K_1 = \hat K_1, K_2 = \hat K_2 \mid \X) \approx (1/T) \sum_{t=1}^T \textbf{1}_{K_1^{(t)} = \hat K_1, K_2^{(t)} = \hat K_2}$.

\section{Illustrations} \label{section:simulations}

\subsection{Two View Setting}

To evaluate the performance of CLIC and the PCDP relative to competitors, we simulate $n=200$ observations under three cases of partition dependence between two views: identical clusterings (case 1), dependent clusterings (case 2), and independent clusterings (case 3). In all cases, we set the number of true clusters in each view to be $K_1 = K_2 = 2$, and simulate $C_1$ via $\Pr(c_{1i} = 1) = \Pr(c_{1i} = 2) = 1/2$. In case 1, we set $c_{2i} = c_{1i}$, and in case 3 we simulate $c_{2i}$ independently using $\Pr(c_{2i} = 1) = \Pr(c_{2i} = 2)=1/2$. For case 2, with probability $2/3$ we sample $(c_{1i}, c_{2i})$ using the scheme in case 1, and with probability $1/3$ we sample $(c_{1i}, c_{2i})$ independently with the scheme in case 3. We focus on the uncorrelated view model \eqref{eq:multiview-model}, but we also explore the correlated view model \eqref{eq:conditional-model} in the Supplementary Material.

When fitting the PCDP, we fix $\gamma_1 = \gamma_2 = 1$, but assume that $\rho \sim \tx{Gamma}(1, 1)$. We run the Gibbs sampler in Section \ref{section:computation} for $30,000$ iterations, remove the first $10,000$ as burn-in, then keep every other iteration, leading to $10,000$ samples from $\pi(C_1, C_2 \mid \X)$. We find that mixing improves using a conditional sampling scheme for the cluster labels. Rather than sampling $(c_{1i}, c_{2i})$ jointly from its full conditional, we first sample $c_{1i}$, then $c_{2i}$ conditional on the value of $c_{1i}$. We also find that this method vastly speeds up posterior computation in comparison to the joint sampler. Mixing and convergence of the PCDP are evaluated using traceplots and effective sample size for $R(C_1,C_2)$.

\begin{figure}
    \centering
    \includegraphics[scale=0.55]{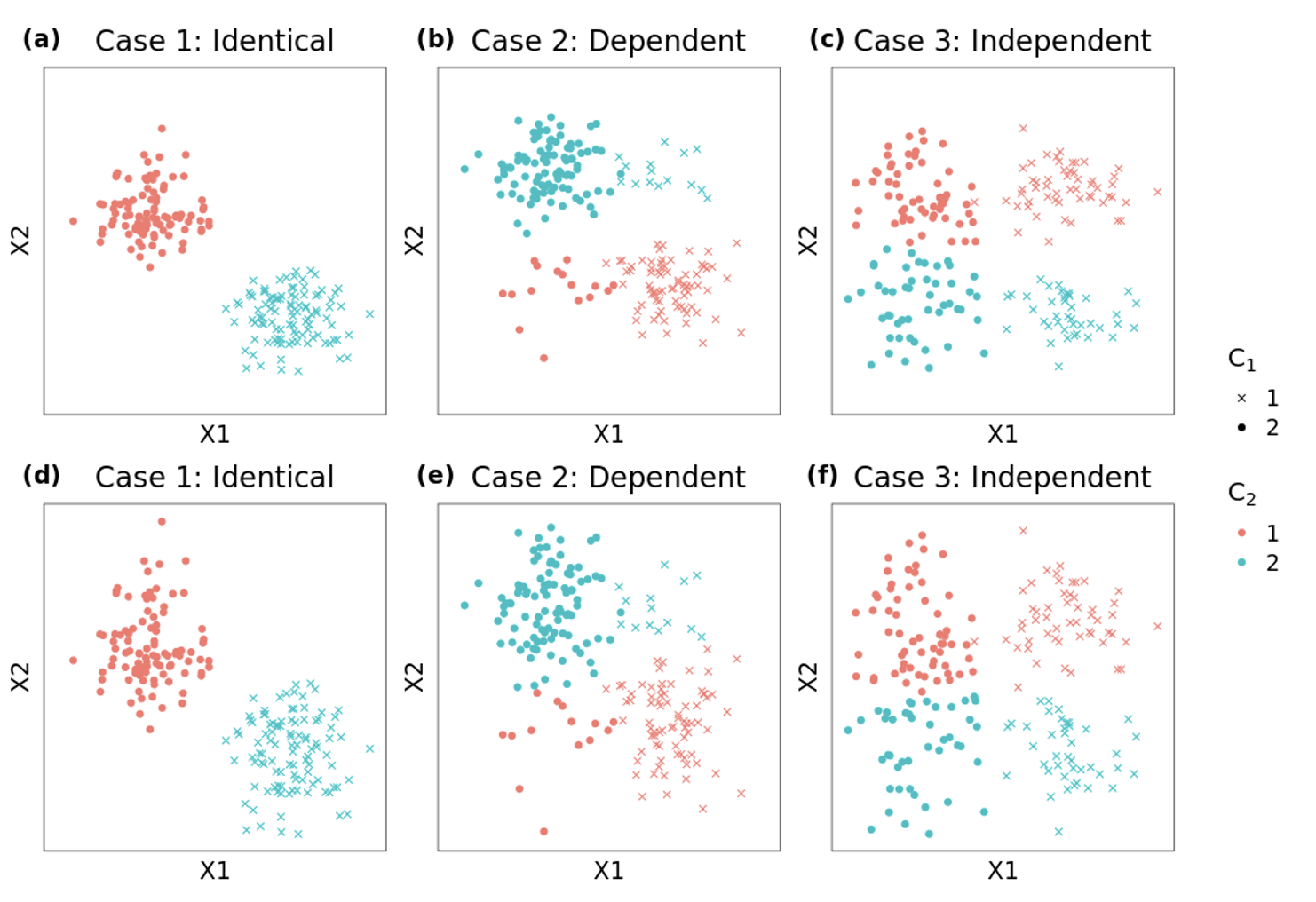}
    \caption{Synthetic data simulated in the two view setting, where shapes and colors correspond to the true values of $C_1$ and $C_2$. Data in plots (a)-(c) are simulated with little overlap between $C_2$ clusters ($\eta^2=0.2$), whereas data in plots (d)-(f) are simulated with higher overlap ($\eta^2=0.45$).}
    \label{fig:multiview-data}
\end{figure}

\begin{table}[ht]
    \centering
    \begin{tabular}{cccccc}
    \toprule
       Overlap & Case & CLIC & t-HDP & EM & IDPs \\
      \multirow{3}{4em}{$\eta^2 \hspace{-0.5mm} = \hspace{-0.5mm} 0.2$} & 1 &  (1.000, 1.000) & (1.000, 1.000) & (1.000, 0.980) & (1.000, 0.980) \\
       & 2 & (0.921, 0.98) & (0.902, 0.98) &   (0.921, 1.000) & (0.921, 1.000) \\
       & 3 & (0.980, 0.921)  & (0.980, 0.921) & (0.980, 0.941) & (0.980, 0.941) \\
       \multirow{3}{4em}{$\eta^2  \hspace{-0.5mm} = \hspace{-0.5mm} 0.45$} & 1 &  (1.000, 1.000) & (1.000, 1.000) & (1.000, 0.000) & (1.000, 0.000) \\
       & 2 & (0.921, 0.921) & (0.883, 0.773) &   (0.921, 0.864) & (0.921, 0.864) \\
       & 3 & (0.980, 0.704)  & (0.980, 0.000) & (0.980, 0.687) & (0.980, 0.687) \\
        \bottomrule
    \end{tabular}
    \caption{The adjusted Rand indices (ARI) between the true $(C_1, C_2)$ and the point estimates computed by CLIC, the t-HDP, independent EM, and two IDPs for synthetic data in the two view setting. CLIC consistently outperforms competitors for inferring $C_2$ at a higher overlap.
    }
    \label{table:aris-marginal}
\end{table}

CLIC is compared to existing methods for clustering uncorrelated views. We evaluate the model in \eqref{eq:multiview-model} across two different overlaps in the sample space for the $C_2$ clusters. That is, we generate data via $(X_i \mid c_{1i}, c_{2i}) \sim \N(X_{1i}; \lb -1\rb^{c_{1i}+1}, 0.2) \times \N(X_{2i};\lb -1 \rb^{c_{2i}}, \eta^2)$, where $\eta^2 \in \lb 0.2, 0.45 \rb$. Smaller values of $\eta^2$ lead to wider separation along the second view. Plots of the synthetic data are displayed in Figure \ref{fig:multiview-data}. We fit a location Gaussian mixture with the finite approximation to the PCDP, i.e. $\pi(X_i \mid \theta_i) = \N(X_{1i}; \theta_{1i}, \sigma_1^2) \times \N(X_{2i}; \theta_{2i}, \sigma_2^2)$, and assume that $H_v(\theta) = \N(\theta; \mu_0, \sigma_0^2)$. We then compute the clustering point estimates corresponding to the telescopic hierarchical Dirichlet process (t-HDP) described in \cite{franzolini2023bayesian}, two independent Dirichlet processes (IDPs), and the EM algorithm \citep{scrucca2023mclust} applied independently to each view. In all Bayesian methods, point estimates are calculated with the VI loss. 

The ARI between the point estimates and the true partitions are displayed in Table \ref{table:aris-marginal}. All methods perform similarly when the clusters along both views are well separated ($\eta^2=0.2$). However, CLIC far more accurately captures $C_2$ in comparison to competitors after increasing the overlap ($\eta^2=0.45$). In case 1, IDPs and the EM algorithm along the marginals completely fail to detect any clustering structure in $C_2$. The t-HDP has similar issues in case 3, as its point estimate allocates all objects to a single cluster. Though all methods detect clustering structure in case 2, CLIC attains the highest ARI with the true $C_2$ clusters, particularly in comparison to the t-HDP. Additionally, CLIC's estimation of $C_1$ is consistent across the two different overlaps, but estimation of $C_1$ in case 2 under the t-HDP worsens when the overlap between the $C_2$ clusters is increased. We also calculated the number of clusters along each view for each method. When $\eta^2 = 0.2$, all methods estimate the correct number of clusters along both views for all three cases. For $\eta^2 = 0.45$, the EM algorithm and IDPs infer one cluster in the second view in case 1. Similarly, the t-HDP results in one cluster along view $2$ in case 3. For all other methods and cases under this overlap, the number of clusters along both views are correctly estimated.

\subsection{Three View Setting}

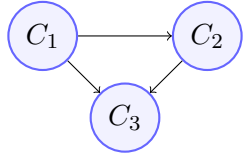
\begin{figure}
    \centering
\begin{tikzpicture}[roundnode/.style={circle, draw=blue!60, fill=blue!5, thick, minimum size=3mm},box/.style = {draw,black,inner sep=10pt,rounded corners=5pt}, node distance = {6mm}]
        \node[roundnode] (c3) {$C_3$};
        \node[roundnode] (c2) [above right = of c3] {$C_2$};
        \node[roundnode] (c1) [above left = of c3] {$C_1$};
        \draw[->] (c1)--(c2);
        \draw[->] (c2)--(c3);
        \draw[->] (c1)--(c3);
    \end{tikzpicture}
    \caption{Directed acyclic graph (DAG) formulation for the joint distribution of the simulated partitions $(C_1,C_2,C_3)$. Note that the DAG is not a tree; $C_3$ has both $C_1$ and $C_2$ as parents.}
    \label{fig:dag}
\end{figure} 

Next, we examine the performance of CLIC and competitors in the three view setting. We set $K_1 = K_2 = K_3=2$ to be the number of true clusters and simulate $n=200$ observations. Here, we only consider the case in which the three partitions are dependent, but not identical. $c_{1i}$ is simulated in the same manner as the two view setting. With probability $1/3$, we set $c_{2i} = c_{1i}$, else we reallocate object $i$ to the other cluster. To simulate $c_{3i}$, we first generate $u_i \sim \tx{Unif}(0,1)$. If $u_i < 1/3$, then $c_{3i} = c_{1i}$, if $1/3 \leq u_i < 2/3$, then $c_{3i} = c_{2i}$, and if $u_i \geq 2/3$, we simulate $c_{3i}$ independently from the other labels using the same cluster probabilities as $c_{1i}$. The directed acyclic graph (DAG) in Figure \ref{fig:dag}  shows the conditional dependencies between all three simulated partitions. Importantly, the dependence structure between $C_1$, $C_2$, and $C_3$ does not correspond to a tree graph, which is how the t-HDP model formulates the joint partition distribution. We simulate the data so that the clusters along $C_1$ and $C_2$ are more separated than the clusters along the third view. The data are generated by $(X_{i} \mid c_{1i}, c_{2i}, c_{3i}) \sim \N(X_{1i}; \lb -1\rb^{c_{1i}+1}, 0.2) \times \N(X_{2i};\lb -1 \rb^{c_{2i}}, 0.2) \times \N(X_{3i}; \lb -1 \rb^{c_{3i}+1}2, 1)$. A plot of the synthetic data is given in Figure \ref{fig:threeview-data}. We fit the finite approximation to the PCDP with three views and set $\pi(X_i \mid \theta_i) = \N(X_{1i}, \theta_{1i}, \sigma_1^2) \times \N(X_{2i}, \theta_{2i}, \sigma_2^2) \times \N(X_{3i}, \theta_{3i}, \sigma_3^2)$, where $H_v(\theta) = \N(\theta; \mu_0, \sigma_0^2)$. We draw the same number of posterior samples and keep the same hyperparameters for CLIC, the t-HDP, and the IDPs as in the two view setting. Clustering point estimates are computed with the VI loss function. ARIs between the point estimates from CLIC and the other existing methods are displayed in Table \ref{table:threeview-results}. CLIC attains the joint highest ARI with the true partitions, but noticeably improves on the t-HDP in estimating $C_3$. This disparity may be explained by the general formulation of dependence in the PCDP: CLIC partitions need not be Markovian or have a tree structure. All methods accurately estimate the number of clusters along the three views.

\begin{figure}
    \centering
    \includegraphics[scale=0.55]{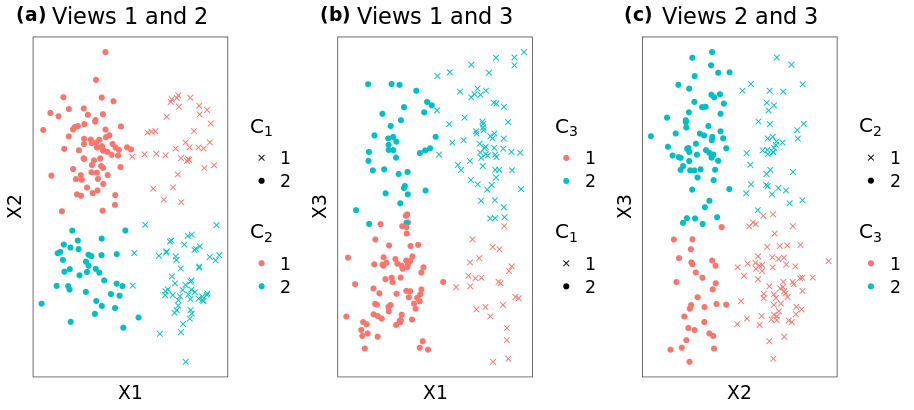}
    \caption{The synthetic data simulated in the three view setting, where shapes and colors correspond to the true values of $C_1$, $C_2$, and $C_3$. 
    }
    \label{fig:threeview-data}
\end{figure}

\begin{table}
    \centering
    \begin{tabular}{ccccc}
    \toprule
         & CLIC & t-HDP & EM & IDPs \\
        ARI with $C_1$ & 0.941 & 0.921 & 0.941 & 0.941 \\
        ARI with $C_2$ & 0.941 & 0.941 & 0.941 & 0.941\\
        ARI with $C_3$ & 0.846 & 0.773 & 0.827 & 0.827 \\
        \bottomrule
    \end{tabular}
    \caption{ARI with the true partitions $C_1$, $C_2$, and $C_3$ for point estimates from CLIC, the t-HDP, the EM algorithm, and two IDPs. CLIC attains the joint highest ARI for all partitions, particularly improving on the t-HDP in estimating $C_3$.}
    \label{table:threeview-results}
\end{table}

\subsection{Varying Sample Size and Dimension}

\begin{table}[ht]
    \centering
    \begin{tabular}{cccccc}
    \multicolumn{6}{c}{$n=100$} \\
    \toprule
       $d_2$  & CLIC & t-HDP & EM & IDPs & DP \\
        \multirow{2}{0.5em}{$2$} & (1.000, 0.919) & (1.000, 0.919) & (1.000, 0.919) & (1.000, 0.919) & (0.835, 0.783) \\
        & (2,2) & (2,2) & (2,2) & (2,2) & (4,4) \\
        \multirow{2}{0.5em}{$10$} & (0.882, 1.000) & (0.960, 1.000) & (0.845, 1.000) & (0.921, 1.000) & (0.866, 0.871)\\
        & (2,2) & (2,2) & (2,2) & (2,2) & (4,4) \\
        \multirow{2}{0.5em}{$25$} & (0.959,1.000) & (1.000,0.143) & (0.959,1.000) & (0.959,1.000) & (0.694,1.000) \\ 
        & (2,2) & (2,10) & (2,2) & (2,2) & (2,2) \\
        \bottomrule
    \end{tabular}
    \begin{tabular}{cccccc}
    \multicolumn{6}{c}{$n=200$} \\
    \toprule
       $d_2$  & CLIC & t-HDP & EM & IDPs & DP \\
        \multirow{2}{0.5em}{$2$} & (0.979,0.939) & (0.959,0.939)& (0.979,0.899) & (0.979,0.899) & (0.843,0.826) \\
        & (2,2) & (2,2) & (2,2) & (2,2) & (4,4) \\
        \multirow{2}{0.5em}{$10$} & (0.827,1.000) & (0.940,1.000) & (0.845,1.000) & (0.827, 1.000) & (0.853, 0.885)\\
        & (2,2) & (2,2) & (2,2) & (2,2) & (4,4) \\
        \multirow{2}{0.5em}{$25$} & (0.980,1.000) & (1.000,0.927) & (0.980,1.000) & (0.980,1.000) & (0.613,1.000) \\
        & (2,2) & (2,4) & (2,2) & (2,2) & (2,2) \\
        \bottomrule
    \end{tabular}
    \begin{tabular}{cccccc}
    \multicolumn{6}{c}{$n=500$} \\
    \toprule
       $d_2$  & CLIC & t-HDP & EM & IDPs & DP \\
         \multirow{2}{0.5em}{$2$} & (0.960,0.960) & (0.976,0.944) & (0.960,0.944) & (0.960,0.944) & (0.825,0.806) \\ 
         &  (2,2) & (2,2) & (2,2) & (2,2) & (4,4) \\
        \multirow{2}{0.5em}{$10$} & (0.906,1.000) & (0.984,1.000) & (0.913,1.000) & (0.906,1.000) & (0.849,0.865)\\
        & (2,2) & (2,2) & (2,2) & (2,2) & (4,4) \\
        \multirow{2}{0.5em}{$10$} & (0.968,1.000) & (1.000,0.966) & (0.968,1.000) & (0.960,1.000) & (0.822,0.813) \\
        & (2,2) & (2,4) & (2,2) & (2,2) & (4,4) \\
        \bottomrule
    \end{tabular}
    \caption{ARIs (top) and number of clusters (bottom) for CLIC and competitors, where $n \in \lb 100, 200, 500 \rb$ and $d_2 \in \lb 2, 10, 25 \rb$. For the single view DP, the partition estimates are the same along both views.}
    \label{table:varying-dim-and-sample-size}
\end{table}

Finally, we also inspect the performance of our approach and competitors by varying the sample size and dimension of the synthetic data. We simulate two highly dependent partitions, where $c_{2i}=c_{1i}$ with probability $8/10$ (otherwise, the two labels are sampled independently). We fix $d_1 = 2$ but vary $d_2 \in \lb 2, 10, 25 \rb$ and $n \in \lb 100,200,500 \rb$. Along the first view, the true cluster means are $\theta_{11}^* = (1,-1)$, $\theta_{12}^* = (-1,1)$. Along the second, these are $\theta_{21}^* = (1,-1)_{d2}$ and $\theta_{22}^* = - \theta_{21}^*$, where $(1,-1)_{d_2}$ is an alternating sequence of $\pm 1$ of length $d_2$. We simulate the data via $X_{vi} \mid \theta_{vi} \sim \N(X_{vi}; \theta_{vi}, 0.4 I_{d_v})$ independently for $v=1,2$. The finite approximation of the PCDP is applied with two independent Gaussian views, $H_{v}(\theta) = \N_{d_v}(\theta; \mu_0, \Sigma_0)$, $\gamma=1$, and $\rho \sim \tx{Gamma}(1,1)$. In addition to the competitors used in the previous illustrations, we also fit a single (i.e., ignorant of the views) DP with a Gaussian likelihood to the full data $\X$.

The resulting ARI values with the true partitions and number of clusters are displayed in Table \ref{table:varying-dim-and-sample-size}. We can see that CLIC generally performs well at estimating the true cluster labels and always results in the correct number of clusters. Our approach excels in the small sample size but high dimension regime ($n=100,d_2=25$), particularly in the estimation of $C_2$. The point estimate of $C_2$ under CLIC perfectly captures the true partition, whereas the ARI for the t-HDP is $0.143$. In general, the t-HDP overestimates the number of clusters along the second view for $d_2 = 25$, resulting in 10 clusters when $n = 100$ and 4 clusters when $n \in \lb 200,500 \rb$. We can also see that CLIC is more robust than a single view DP for estimating both partitions. When $d_2=25$ and $n \in \lb 100, 200 \rb$, the point estimate for the DP perfectly corresponds to the $C_2$ clustering, i.e. the view with dominating dimension, though this tendency is not present in the DP for large $n$. Finally, it is clear that CLIC can outperform the independent view methods (i.e., EM and IDPs) for smaller values of $d_2$.

\section{Example: The Collaborative Perinatal Project} \label{section:application}

The \cite{longnecker2001association} dataset is a subset of children studied in the US Collaborative Perinatal Project (CPP) \citep{niswander1972women}. The CPP enrolled pregnant women born between 1959 and 1966, and measured the participants' blood pre-delivery, at delivery, and postpartum. The children of the participants were then monitored for a period of seven years after delivery for neurodevelopmental outcomes. From 1997 to 1999, the \cite{longnecker2001association} study assessed the concentration of DDE, a metabolite of the insecticide DDT, present in blood samples taken from the third trimester using laboratory assays. They were interested in inferring the relationship between DDE concentration and two primary outcomes: preterm delivery and being small-for-gestational age (SGA), the latter term describing newborns with smaller weight than what is considered normal. In total,  \cite{longnecker2001association} obtained complete information on $2,380$ children, where $361$ were classified as being delivered preterm and $221$ were small-for-gestational-age. 

Following standard procedure in reproductive epidemiology, the two primary outcomes were determined by categorizing birth weight and gestational age. For instance, preterm, term, and post-term births correspond to gestational age at delivery less than 37 weeks, between 37 and 42 weeks, and after 42 weeks, respectively; preterm births are often further classified as very, moderately, or late preterm \citep{loftin2010late}. A baby's birth weight is considered to be low, or ``premature", if it does not exceed 2.5 kg \citep{hughes2017low}. There are multiple ways to classify newborns as SGA using birth weight and gestational age, and these generally are study-specific and use percentiles \citep{schlaudecker2017small}. The issue is further complicated by the fact that the start of gestation, and therefore gestational age, is difficult to measure in practice. Therefore, we are interested in inferring \textit{naturally-occurring} subgroups along both birth weight and gestational age, motivating a multiview approach, where view 1 is birth weight and view 2 is gestational age. 

The data are collected in the matrix $\X = (\X_1, \X_2)$, where $X_{1i}$ is the normalized birth weight and $X_{2i}$ is the normalized gestational age for child $i$. We assume that $X_i \mid \theta_i \sim \N(X_{1i}; \theta_{1i}, \sigma_1^2 ) \times \N(X_{2i}; \theta_{2i}, \sigma_2^2)$, $\theta_i \mid P \sim P$, and $P \sim \tx{PCDP}$, where the base measures are $H_v(\theta) = \N(\theta; 0, 1)$ and $\gamma_1 = \gamma_2 = 1$. We take $L_1=L_2=5$ and $\pi(\rho = u ) \propto 1$ for all $u \in U$, where $U$ is taken to be a grid from $10^{-2}$ to $150$ in intervals of length $0.5$. That is, we implement the griddy-Gibbs sampler discussed in Section \ref{section:computation}. We only consider children born before 42 weeks (i.e., full term newborns), then take a random sample of size $n=1,000$. The Gibbs sampler is run for $100,000$ iterations with the first $10,000$ taken as burn-in. We take every fifth iteration, leading to $18,000$ samples from $\pi(C_1,C_2 \mid \X)$. The griddy-Gibbs sampler leads to better MCMC diagnostics compared to the \cite{escobar1995bayesian} semi-conjugate prior for $\rho$. Convergence of the partitions and $\rho$ was rapid and mixing was good based on trace plots and effective sample size. Partition point estimates $\widehat C_1$ and $\widehat C_2$ are calculated by minimizing the VI loss.

\begin{figure}
    \centering
    \includegraphics[scale=0.55]{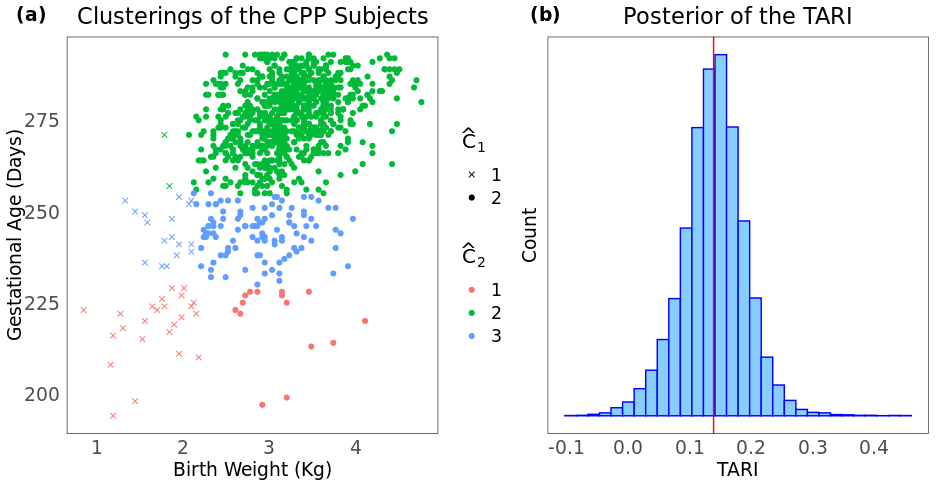}
    \caption{On the left, gestational age and birth weight for the children studied in \cite{longnecker2001association} along with the clustering point estimates, where shapes correspond to groups in birth weight and colors correspond to groups in gestational age. On the right, the posterior distribution of the TARI between $C_1$ and $C_2$, where the red line is the estimated posterior mean ($0.139$). Note that similar to the ARI, the TARI can take negative values with positive probability. }
    \label{fig:longnecker-data-and-TARI}
\end{figure}

The clustering point estimates are displayed in Figure \ref{fig:longnecker-data-and-TARI}(a). There are two birth weight clusters, cluster 1 (with $43$ children) and cluster 2 (with $957$ children). The maximum birth weight of children in cluster $1$ of $\widehat C_1$ was $2.183$ kg, meaning that all children in cluster $1$ would be considered to be premature. $97.7\%$ of the children in this cluster are born pre-term; and that proportion for cluster 2 of $\widehat C_1$ is merely $16.1\%$. In addition, Pearson's chi-squared test rejects the null independence hypothesis between $\widehat C_1$ and preterm delivery incidence. Interestingly, $\widehat C_2$ has three clusters which do not correspond exactly to the preterm delivery cut-offs. While all $804$ full term children were allocated into cluster 2 of $\widehat C_2$, the remaining $196$ preterm subjects were split between the three clusters, with $19.9\%$ in cluster 1, $16.3$ \% in cluster $2$, and $63.4\%$ in cluster 3. In fact, all of the subjects allocated to clusters 1 and 3 of $\widehat C_2$ were delivered pre-term. Similar to $\widehat C_1$, the chi-squared test for independence between $\widehat C_2$ and pre-term delivery rejects the null hypothesis.

\begin{table}[ht]
\centering
\begin{tabular}{ccccc}
  & $\hat{C}_{21}$ & $\hat{C}_{22}$ & $\hat{C}_{23}$ &   \\ 
$\hat{C}_{11}$ &  24 &   2 &  17 &  43 \\ 
   $\hat{C}_{12}$ &  15 & 834 & 108 & 957 \\ 
     & 39 & 836 & 125 & 
\end{tabular}
\caption{Contingency table between the clustering point estimates $\widehat C_1$ and $\widehat C_2$ for the \cite{longnecker2001association} data, along with marginal cluster sizes along the rows and columns. }
\label{table:longnecker-contingency}
\end{table}

Along with the clustering point estimates, we find evidence that groups based on birth weight and gestational age are dependent. We quantify the dependence between the partitions using the TARI. The TARI posterior distribution is displayed in Figure \ref{fig:longnecker-data-and-TARI}(b), and has posterior expectation equal to $0.139$ and 95\% credible interval $[0.029, 0.239]$. An interesting insight provided by the multiview approach is the ability to infer premature subgroups in each gestational age cluster by inspecting the contingency table for the point estimates (Table \ref{table:longnecker-contingency}). Here, we can see, for example, that newborns in $\hat C_{21}$ are more likely to be born prematurely than newborns in the other birth weight clusters, and that this probability rapidly diminishes with increasing gestational age. In addition, the separation between the $\hat C_1$ clusters is far larger in $\hat C_{21}$ than $\hat C_{22}$ and $\hat C_{23}$, indicating that birth weight outcomes are most polarized for earlier births.

\section{Discussion and Extensions} \label{section:discussion}

This article proposes a novel method for Bayesian inference on dependent partitions arising from multiview data. Our approach relies on the PCDP, a nonparametric prior that generates dependent random distributions by centring a DP on a random product measure. The induced $\CLIC(\rho, \gamma_1, \gamma_2)$ distribution has several appealing theoretical properties, such as an independence centered CRP and a closed form expression of the ERI. The probability mass function of the $\CLIC(\rho, \gamma_1, \gamma_2)$ distribution is an MEPPF, a novel generalization of the EPPF that describes the joint distribution of random partitions in terms of their contingency table. In addition, we derive a Gibbs sampler for a finite approximation to the PCDP that can model both uncorrelated and correlated views. Point estimation and uncertainty quantification of the partitions is achieved using standard procedures in the Bayesian clustering literature. A single parameter, $\rho$, controls perturbations from the independent clustering model, and the degree of dependence between the partitions is inferred using the Rand index (or the normalized TARI). In simulations and an application to the \cite{longnecker2001association} study, we show that CLIC can capture a wide spectrum of partition dependence.

Our method contributes to the literature on multiview clustering. Interpretability, simplicity, and the ability to characterize uncertainty in clustering motivate the use of CLIC in practice. However, in some applications collecting multiview data, the focus is on inferring a single grouping of the objects, e.g. when clusters are used to make decisions, such as treatment assignment. Consensus clustering approaches \citep{lock2013bayesian, xanthopoulos2014} may be more appropriate in this regime. In addition, we assume that the views are known and correspond to well-defined groups of features. However, in some cases, we do not know in advance which features should be assigned to which views; in such cases, clustering features into views may be one of the main goals. There is a related literature on bi-clustering \citep{castanho2024biclustering}, which seeks to simultaneously learn clusters of samples and features, but such approaches do not allow different sample clusters for each feature cluster. Finally, some applications require that the derived clusterings are highly dissimilar; this is the main aim of multiple clustering \citep{bailey2014alternative}. Although higher values of $\rho$ encourage independent partitions a priori, our methodology generally does not prioritize dissimilar partitions. 

In comparison to existing methods for Bayesian multiview clustering, our approach does not assume a directional relationship between $C_1$ and $C_2$. Instead, our methodology focuses on describing the joint distribution of $(C_1,C_2)$ using the innovative idea of a cross-partition and, in particular, the MEPPF. The CLIC distribution is one example of a joint partition model that admits an MEPPF, but it is interesting to examine what other joint partition models fall under this framework. Though our construction relies on a hierarchical process for sampling a random measure, it may be the case that an MEPPF can be induced without a data augmentation procedure. This formulation would be especially useful for computation. Modeling $\rho$ as random leads to a fully Bayesian model and additionally leads to vastly improved mixing over setting it equal to a constant, but sampling from the full conditional of $\rho$ requires imputation of the marginal latent probability vectors $\tilde q_1, \tilde q_2$. As in the HDP, the addition of these latent vectors can lead to slow mixing in certain situations. The closed form expression for the distribution of the cross-partition in \eqref{eq:joint-EPPF} and the notion of root partitions provide a promising direction for deriving a marginal MCMC sampler, though sampling from $\pi(T_1, T_2 \mid C_1, C_2, \rho, \gamma_1, \gamma_2)$ is not straightforward. 

The PCDP can be extended to address several tasks that commonly arise in unsupervised learning. For example, it is often of interest to simultaneously infer a random partition while selecting features that distinguish its clusters \citep{hancer2020survey}. A common strategy is to introduce latent discrimination indicators \citep{tadesse2005bayesian, kim2006variable} $\alpha_{vj_v} \in \lb 0,1 \rb$ for each feature and view. Typically, the indicators are assumed to be independent and identically distributed $\tx{Bernoulli}(\omega)$ random variables. Let $\theta_i^{(1)}$ be the vector of atoms for which $\alpha_{vj_v}=1$ and $\theta_i^{(0)}$ be the analogous vector for which $\alpha_{vj_v}=0$; let $\alpha^{(1)}$ and $\alpha^{(0)}$ be similar vectors for the indicators. Inter-view variable selection can be accomplished by modeling $\theta_i^{(0)} \mid \alpha^{(0)} \sim W$, where $W$ is a fixed non-atomic measure, then modeling $\theta_i^{(1)}$ conditional on $\alpha^{(1)}$ with the PCDP. However, it may be the case that we want to adjust for the different views when selecting variables, e.g., when one view is of far higher dimension than the other. To implement this intra-view approach, we concatenate the indicators into $\alpha_v = (\alpha_{vj_v})_{j_v=1}^{d_v}$ for all $v=1, \dots, V$. We model the members of each $\alpha_v$ as independent $\tx{Bernoulli}(\omega_v)$ variables, where $\omega_v$ can be calibrated to match the desired sparsity for the $v$th view. If $\theta_{vi}^{(0)}$ are the atoms in the $v$th view for which $\alpha_{vj_v}=0$, we draw $\theta_{vi}^{(0)} \mid \alpha_{v}^{(0)} \sim W_v$ for a fixed measure $W_v$. Alternatively, one could adapt shrinkage priors developed for variable selection \citep{yau2011hierarchical,malsiner2016model} into intra-view and inter-view approaches with the PCDP as the prior on the partitions. 

Finally, an intriguing prospect of the PCDP is extending its hierarchical formulation to incorporate graphical dependence structure between the views. For example, for $V=3$ views, the directed acyclic graph (DAG) for the partitions could be $C_1 \to C_2 \to C_3$, in which case $\pi(c_{1i} = k_1, c_{2i} = k_2, c_{3i} = k_3 \mid  p) = p_{k_1 k_2}^{(12)}   p_{k_3 \mid k_2}^{(3 \mid 2)}$, where $ p_{k_v \mid k_{v-1}}^{(v \mid v-1)}$ is the conditional probability of the $v$th label given the $(v-1)$th label for clusters $k_v$ and $k_{v-1}$. In this case, the PCDP could be used to model the pairwise joint distributions between views one and two, and between views two and three. This would allow separate parameters, say $\rho_{12}$ and $\rho_{23}$, to control concentration towards pairwise independent models. An alternative strategy is to center a random measure on the DAG, i.e. $P \mid \rho, Q \sim \DP(\rho, Q_{12} \times Q_{3 \mid 2})$, where $Q_{12}$ is a random joint distribution between the first two views and $Q_{3 \mid 2}$ is a random conditional measure between views two and three. When $\rho \to \infty$, the DAG holds, and hence $\rho$ can be used to test the validity of the graphical model. 

\section*{Acknowledgements}
This work was supported by the National Institutes of Health (NIH) under Grants 1R01AI155733, 1R01ES035625, and 5R01ES027498; the United States Office of Naval Research (ONR) Grant N00014-21-1-2510; the European Research Council (ERC) under Grant 856506; and Merck \& Co., Inc., through its support for the Merck BARDS Academic Collaboration. The authors thank Amy H. Herring for helpful comments.

\bibliography{sources}

\appendix 

\section*{Supplementary Material}

Proofs of all results and extended details on simulations and the application to the CPP data are available in the Supplementary Material. Code for the simulations and application is available at \url{https://github.com/adombowsky/clic}.

\section{Proofs and Additional Results}

\subsection{Identifiability of Concentration Parameters}
In this section, we verify that the key hyperparameters $(\rho, \gamma_1, \gamma_2)$ are identifiable in the joint prior for $(P, Q_1, Q_2)$. Let $S_{v1}, \dots, S_{vm_v}$ be a partition of $\Theta_v$ and assume that $H_v(S_{vk_v})>0$ for all $k_v=1, \dots, m_v$. By definition of the DP,
\begin{equation} \label{eq:q-measure-of-partition}
    (Q_v(S_{v1}), \dots, Q_v(S_{vm_v})) \mid \gamma_v \sim \Dir(\gamma_v H_v(S_{v1}), \dots, \gamma_v H_v(S_{vm_v})).
\end{equation}
This Dirichlet distribution is a one parameter exponential family, with sufficient statistic $\sum_{u_v=1}^{m_v} H_v(S_{vm_v}) \log Q_v(S_{vm_v})$ and canonical parameter $\gamma_v$. Hence, $\gamma_v$ is identifiable in the prior for $Q_v$. Moreover, let $S_1, \dots, S_m$ be a partition of $\Theta = \Theta_1 \times \Theta_2$ and $Q$ be \textit{any probability measure} on $\Theta$ with $Q(S_k) >0$ for all $k=1, \dots, m$. By definition,
\begin{equation} \label{eq:p-measure-of-partition}
    (P(S_1), \dots, P(S_m)) \mid \rho, Q \sim \Dir(\rho Q(S_1), \dots, \rho Q(S_m)), 
\end{equation}
and so by a similar argument for $Q_v$, $\rho$ is identifiable in the prior for $P \mid Q$.

Next, let $\lb \lb S_{v1}, \dots, S_{vm_v} \rb \rb_{v=1,2}$ be partitions of $\Theta_1$ and $\Theta_2$ with $H_v(S_{vk_v}) > 0$ for all $k_v$ and $v$, and create a partition of $\Theta$ with $\lb S_{1k_1} \times S_{2k_2} \rb_{ u_1, u_2}$. We now define a joint prior density over (a) a measure of the joint partition $\lb S_{1u_1} \times S_{2u_2} \rb_{ u_1, u_2}$, (b) a measure of the marginal partition along $\Theta_1$, and (c) a measure of the marginal partition along $\Theta_2$. Let $p$ belong to the $m_1m_2-1$ simplex, $q_1$ belong to the $m_1-1$ simplex, and $q_2$ belong to the $m_2-1$ simplex. The notation $\pi(p \mid \rho, q_1, q_2)$ denotes the density function of \eqref{eq:p-measure-of-partition} and $\pi(q_v \mid \gamma_v)$ denotes the density function of \eqref{eq:q-measure-of-partition} for $v=1,2$. We occasionally collect all hyperparameters into the vector $\varphi = (\rho, \gamma_1, \gamma_2)$. Hence, $\pi(p, q_1, q_2 \mid \varphi) = \pi(p \mid \rho, q_1, q_2) \pi(q_1 \mid \gamma_1) \pi(q_2 \mid \gamma_2)$. 

Suppose that $\pi(p, q_1, q_2 \mid \varphi) = \pi(p, q_1, q_2 \mid \varphi^\prime)$ for all $p, q_1, q_2$ in their respective supports. This is equivalent to 
\begin{equation}
    \forall p, q_1, q_2, \quad \pi(p \mid \rho, q_1, q_2) \pi(q_1 \mid \gamma_1) \pi(q_2 \mid \gamma_2) = \pi(p \mid \rho^\prime, q_1, q_2) \pi(q_1 \mid \gamma_1^\prime) \pi(q_2 \mid \gamma_2^\prime). \label{eq:pi-p-q1-q2}
\end{equation}
Thus, for any $q_1, q_2$, 
\begin{gather}
    \int \pi(p \mid \rho, q_1, q_2) \pi(q_1 \mid \gamma_1) \pi(q_2 \mid \gamma_2) \diff p = \int \pi(p \mid \rho^\prime, q_1, q_2) \pi(q_1 \mid \gamma_1^\prime) \pi(q_2 \mid \gamma_2^\prime) \diff p \nonumber \\
    \implies  \forall q_1, q_2, \quad \pi(q_1 \mid \gamma_1) \pi(q_2 \mid \gamma_2) = \pi(q_1 \mid \gamma_1^\prime)\pi(q_2 \mid \gamma_2^\prime). \label{eq:pi-q1-q2}
\end{gather}
Similarly, for any $q_1$,
\begin{gather}
    \int \pi(q_1 \mid \gamma_1) \pi(q_2 \mid \gamma_2) \diff q_2 = \int \pi(q_1 \mid \gamma_1^\prime)\pi(q_2 \mid \gamma_2^\prime) \diff q_2 \nonumber \\
    \implies \forall q_1, \quad \pi(q_1 \mid \gamma_1) = \pi(q_1 \mid \gamma_1^\prime). \label{eq:pi-q1}
\end{gather}
As $\gamma_1$ is identifiable, \eqref{eq:pi-q1} implies $\gamma_1 = \gamma_1^\prime$. \eqref{eq:pi-q1-q2} simplifies to $\pi(q_2 \mid \gamma_2) = \pi(q_2 \mid \gamma_2^\prime)$ for all $q_2$ and so $\gamma_2 = \gamma_2^\prime$. By a similar argument for \eqref{eq:pi-p-q1-q2}, $\rho = \rho^\prime$.

\subsection{Proof of Theorem \ref{thm:thm:polya-urn}}

\begin{proof}
In the independence centered CRP, the customer can either choose an entre\'e and table from the contingency table, or they can make their choices independently. Denote $r^{(i)}$ to be the number of times that customers $1, \dots, i-1$ made independent choices. To make the CRP metaphor more concrete, we introduce the notion of pseudo-unique values $\vartheta_1^*, \dots, \vartheta^*_{r^{(i)}}$ and view-specific pseudo-unique values $\vartheta^*_{v1}, \dots, \vartheta^*_{vr^{(i)}}$, defined by $\vartheta_l^* \mid Q  \sim Q$ for all $l=1, \dots, r^{(i)}$.

Recall the stick-breaking process representation of the DP \citep{sethuraman1994constructive},
    \begin{align} \label{eq:Q-stick}
        Q_v & = \sum_{k_v=1}^\infty q_{vk_v} \delta_{\theta_{vk_v}^*}, & q_{vk_v} & = \beta_{vk_v} \prod_{m_v<k_v} (1 - \beta_{vm_v}), & \beta_{vk_v} & \sim \tx{Beta}(1, \gamma_v), & \theta_{vk_v}^* & \sim H.
    \end{align}
Since the $Q_v$ are discrete and $Q$ is the product measure, we have that for any $\theta = (\theta_1, \dots, \theta_V) \in \Theta$,
\begin{gather}
    Q(\theta) \overset{a.s.}{=} \sum_{k_1=1}^\infty \cdots \sum_{k_V=1}^\infty \lb  \prod_{v=1}^V q_{vk_v} \rb \delta_{\theta_v^*}(\theta) \nonumber \\ 
    \overset{a.s.}{=} \sum_{k_1=1}^\infty \cdots \sum_{k_V=1}^\infty \prod_{v=1}^V q_{vk_v} \delta_{\theta_{vk_v}^*}(\theta_v) \overset{a.s.}{=} \prod_{v=1}^V \lb \sum_{k_v=1}^\infty q_{vk_v} \delta_{(\theta_{vk_v}^*)} \rb = \prod_{v=1}^V Q_v(\theta_v),
\end{gather}
by Fubini's theorem. Hence, by independence we have that
\begin{equation*}
    \E[Q \mid \vartheta^*_{1}, \dots, \vartheta^*_{r^{(i)}}] = \bigtimes_{v=1}^V \E[Q_v \mid \vartheta^*_{1}, \dots, \vartheta^*_{r^{(i)}}] = \bigtimes_{v=1}^V \E[Q_v \mid \vartheta^*_{v1}, \dots, \vartheta^*_{vr^{(i)}}].
\end{equation*}

Observe that each term in the product measure above is just the posterior expectation of a Dirichlet process (see, for example \cite{gelman2013bayesian}),
\begin{equation} \label{eq:Q_u-Polya}
    \E[Q_v \mid \vartheta_{v1}^*, \dots, \vartheta_{vr^{(i)}}^*] = \frac{1}{\gamma_v + r^{(i)}} \sum_{l_v=1}^{r^{(i)}} \delta_{\vartheta_{vl_v }^*} + \frac{\gamma_v}{\gamma_v + r^{(i)}} H_v.
\end{equation}
Under this sampling scheme, we can now see that, provided $H_v$ is non-atomic, there will be ties amongst the pseudo-unique values $\vartheta_{v1}^*, \dots, \vartheta_{vr^{(i)}}^*$ into unique values $\theta_{v1}^*, \dots, \theta_{vK_v^{(i)}}^*$ where $K_v^{(i)} \leq r^{(i)}$. We then denote $r_{vk_v} = \sum_{l_v=1}^{r^{(i)}} \textbf{1}(\vartheta_{vl_v}^* = \theta_{v k_v}^*)$ for all $k_v=1,\dots,K_v$, which counts the number of times that the first $i-1$ customers chose an occupied table or already served dish given that they made the choices independently. This also implies that $\sum_{k_v=1}^{K_v^{(i)}} r_{vk_v} = r^{(i)}$. We can then rewrite \eqref{eq:Q_u-Polya} as
\begin{equation} \label{eq:Q_u-polya-unique}
    \E[Q_v \mid \theta_{v1}^*, \dots, \theta_{vK_v}^*] = \frac{1}{\gamma_v + r^{(i)}} \sum_{k_v=1}^{K_v^{(i)}} r_{vk_v} \delta_{\theta_{vk_v}^*}  + \frac{\gamma_v}{\gamma_v + r^{(i)}} H_v.
\end{equation}

Now, we address the predictive distribution of $\theta_i$. Conditional on $Q$ and $\vartheta_1^*, \dots, \vartheta_{r^{(i)}}^*$, the Blackwell-Macqueen urn scheme \citep{blackwell1973ferguson} holds (e.g., as in the HDP in \cite{teh2006sharing}), resulting in
\begin{equation} \label{eq:Polya-given-Q}
    \theta_i \mid \vartheta_1^*, \dots, \vartheta_{r^{(i)}}^*, \rho, Q \sim \frac{1}{\rho + i - 1} \sum_{l_1, \dots, l_V} w_{l_1 \cdots l_V} \delta_{(\vartheta^*_{1l_1}, \dots, \vartheta^*_{Vl_V})} + \frac{\rho}{\rho+i-1}Q,
\end{equation}
where the sum is taken over $l_1 \in [r^{(i)}], \dots, l_V \in [r^{(i)}]$ and $w_{l_1 \cdots l_V} = \sum_{j=1}^{i-1}\textbf{1}(\theta_{j1}=\vartheta_{1l_1}^*, \dots, \theta_{jV} = \vartheta^*_{Vl_V})$, and $\sum_{k_1, \dots, k_v}w_{k_1 \cdots k_v} = i-1$. As mentioned previously, there are groups amongst the pseudo-unique values, which we count with
\begin{equation*}
    n_{k_1 \cdots k_v} = \sum_{l_1, \dots, l_V} \textbf{1}(\vartheta_{1l_1} = \theta_{1k_1}^*, \dots, \vartheta_{Vl_V}^* = \theta_{Vk_v}^*)w_{l_1 \cdots l_V},
\end{equation*}
which correspond to the entries in the contingency array. Denote $\bs r_v = (r_{v1}, \dots, r_{vK_v^{(i)}})$ for all $v=1, \dots, V$. Using \eqref{eq:Q_u-Polya} and \eqref{eq:Q_u-polya-unique}, we can marginalize out $Q$ from \eqref{eq:Polya-given-Q} to obtain
\begin{equation*}
    \theta_i \mid \theta^{(i)}, \rho \sim \hspace{-4mm} \sum_{k_1, \dots, k_V}   \frac{n_{k_1 \cdots k_v}}{\rho + i - 1} \delta_{(\theta_{1k_1}^*, \dots, \theta_{Vk_V}^*)} + \frac{\rho}{\rho + i - 1} \bigtimes_{v=1}^V \lb  \sum_{k_v=1}^{K_v^{(i)}} \frac{r_{vk_v}}{\gamma_v + r^{(i)}} \delta_{\theta_{vk_v}^*}  + \frac{\gamma_v}{\gamma_v + r^{(i)}} H_v \rb.
\end{equation*}

 \end{proof}

\subsection{Proof of Proposition \ref{prop:convergence}}
\begin{proof}
First, we show (i). Let $\epsilon>0$ and assume that $\rho$ is fixed, i.e., we will show all results by conditioning on $\rho$, though we will not use the conditional notation. Using Markov's inequality,
\begin{gather} \label{eq:markov-inequality}
    \pi( |P(S) - Q(S)| > \epsilon \mid Q) \leq \frac{\Var[P(S) \mid Q]}{\epsilon^2} = \frac{ Q(S) \lb 1 - Q(S) \rb}{\epsilon^2(\rho+1)} \leq \frac{1}{4\epsilon^2(\rho+1)}.
\end{gather}
Next, we take the expectation of both sides of \eqref{eq:markov-inequality} with respect to the prior for $Q$. Since the right hand side does not depend on $Q$, we have that $\pi(|P(S) - Q(S)| > \epsilon) \leq 1/(4\epsilon^2 (\rho + 1))$ for any fixed value of $\rho>0$. If we construct a countable sequence of $\rho$ values with an infinite limit, this shows that for any $\epsilon>0$, $\lim_{\rho \to \infty} \pi(|P(S) - Q(S)| > \epsilon) = 0$.

For (ii), observe that $Q_v(S_v) \overset{p}{\longrightarrow} H_v(S_v)$ as $\gamma_v \to \infty$ for $v=1,2$. Fix $\lambda \in [0,1]$ to be any continuity point for the CDF of $P(S)$.
Recall that $P(S) \mid \rho, Q(S) \sim \tx{Beta}(\rho Q(S), \rho Q(S^c))$.
Since the partial derivatives of the incomplete beta function exist for all shape parameters \citep{ozcag2008partial}, the continuous mapping theorem implies convergence of the CDFs in prior probability,
\begin{equation}
    \Pi(\lambda \mid Q_1(S_1), Q_2(S_2)) \cvp \Pi(\lambda \mid H_1(S_1), H_2(S_2)).
\end{equation}
Using the dominated convergence theorem applied to $\pi \lb P(S) \leq \lambda \mid Q_1(S_1), Q_2(S_2) \rb$, we get that $P(S) \overset{d}{\longrightarrow} P_0(S)$.
\end{proof}

\subsection{Proof of Proposition \ref{prop:finite-approximation}}
\begin{proof}
    \cite{ishwaran2002exact} shows that as $L_v \to \infty$, $\int h_v(\theta_v) \diff \tilde{Q}_v(\theta_v) \overset{d}{\to} \int h_v(\theta_v) \diff Q_v(\theta_v)$ for all measurable $h_v$. Using Fubini's theorem, we have that
    \begin{equation}
        \int \int h(\theta) \diff\tilde{Q}(\theta) = \int h_1(\theta_1) \diff \tilde{Q}_1(\theta_1)  \int h_2(\theta_1) \diff \tilde{Q}_2(\theta_2).
    \end{equation}
    Recall that $\tilde{Q}_1 \ind \tilde{Q}_2$, $\tilde{Q}_1 \ind Q_2$, $\tilde{Q}_2 \ind Q_1$, and $Q_1 \ind Q_2$. Therefore, by the continuous mapping theorem,
    \begin{equation*}
        \int h_1(\theta_1) \diff \tilde{Q}_1(\theta_1)  \int h_2(\theta_1) \diff \tilde{Q}_2(\theta_2) \overset{d}{\longrightarrow} \int h_1(\theta_1) \diff Q_1(\theta_1) \int h_2(\theta_1) \diff Q_2(\theta_2),
    \end{equation*}
    and the result follows by applying Fubini's theorem to the right hand side of the above. 
    
\end{proof}

\subsection{Proof of Theorem \ref{thm:joint-EPPF}}
\begin{proof}
The result largely follows by adapting the proof of Theorem 1 in \cite{camerlenghi2018bayesian}, which derives the EPPF resulting from an HDP. Let $\theta_1 \neq \cdots \neq \theta_K$ be fixed, where $\theta_k = (\theta_{1k}, \theta_{2k})$. Consider the function
\begin{equation*}
    M_{\n}(d \theta_1, \dots, d \theta_K) = \E \lb \prod_{k_1,k_2} P^{n_{k_1k_2}}(d (\theta_{1k_1}, \theta_{2k_2})) \rb.
\end{equation*}
The connection between $M_{\n}(d \theta_1, \dots, d \theta_K)$ and $\Pi(C_1,C_2 \mid \rho, \gamma_1, \gamma_2)$ is that the MEPPF is equal to the integral of the above with respect to $\theta_1 \neq \dots \neq \theta_K$. Suppose now that we create an $\epsilon$-rectangle around each $\theta_k$, denoted by $B_{k_1,k_2,\epsilon} = B_{1k_1, \epsilon} \times B_{2k_2, \epsilon}$, where $B_{vk_v, \epsilon} = \lb x \in \Theta_v: \norm{\theta_{vk_v}-x} < \epsilon \rb$, and $\epsilon$ is small enough so that $B_{k_1,k_2,\epsilon}$ are disjoint for all $k_1,k_2$. It follows then that
\begin{gather*}
    M_{\n}\left( \bigtimes_{k_1,k_2} B_{k_1,k_2,\epsilon} \right) =  \E \lb \prod_{k_1,k_2} P^{n_{k_1k_2}}(B_{k_1,k_2,\epsilon}) \mid Q_1,Q_2 \rb \\
    = \frac{1}{\Gamma(n)} \E \int_{0}^{\infty} u^{n-1} e^{- \rho \Psi(u) Q(\Theta_\epsilon^c) } \prod_{k_1,k_2} \lb (-1)^{n_{k_1k_2}} \frac{d^{n_{k_2 k_2}}}{d u^{n_{k_1 k_2}}} e^{-\rho \Psi(u) Q(B_{k_1,k_2,\epsilon})} \rb \diff u,
\end{gather*}
where $\Theta_\epsilon^c = \Theta \setminus \bigcup_{k_1,k_2} B_{k_1,k_2,\epsilon}$ and for any $u > 0$, $\Psi(u) = \int_{0}^{\infty} \lb 1 - e^{-u \nu}  \rb \nu^{-1} e^{-\nu} d \nu$. By applying equation (A2) in \cite{camerlenghi2018bayesian}, we have that 
\begin{gather}
    M_{\n}\left( \bigtimes_{k_1,k_2} B_{k_1,k_2,\epsilon} \right) \nonumber \\
    = \frac{1}{\Gamma(n)} \int_{0}^{\infty} u^{n-1}e^{- \rho \Psi(u)} \sum_{r \in \mathcal R} \rho^{r} \left( \E \prod_{k_1, k_2} Q^{r_{k_1 k_2}}(B_{k_1,k_2,\epsilon}) \xi_{n_{k_1k_2},r_{k_1k_2}}(u) \right) \diff u, \label{eq:camerlengthi}
\end{gather}
where for any two integers $m,n$,
\begin{equation*}
    \xi_{n, m}(u) = \frac{1}{m!} \sum_{w \in W} \binom{n}{w_1 \cdots w_m} \frac{\prod_{l=1}^m \Gamma(w_l)}{(u+1)^{n}} ,
\end{equation*}
and the summation is over $W = \lb w=(w_1, \dots, w_m): w_l \in \mathbb{N}, \sum_{l=1}^m w_l = n \rb$. Observe that, under the assumption that $\mathcal{S}_1, \mathcal{S}_2$ are Borel algebras, with probability equal to 1,
\begin{gather*}
    \prod_{k_1,k_2} Q^{r_{k_1k_2}}(B_{k_1,k_2,\epsilon}) = \prod_{k_1,k_2} \lb  Q_1^{r_{k_1k_2}}(B_{1k_1,\epsilon}) Q_2^{r_{k_1,k_2}}(B_{2k_2,\epsilon}) \rb \\
    = \prod_{k_1} Q_1^{r_{1k_1}}(B_{1k_1,\epsilon}) \prod_{k_2} Q_2^{r_{2k_2}}(B_{2k_2, \epsilon}).
\end{gather*}
By independence of $Q_1$ and $Q_2$,
\begin{gather*}
    \E \prod_{k_1,k_2} Q^{r_{k_1 k_2}}(B_{k_1,k_2,\epsilon}) = \E \prod_{k_1} Q_1^{r_{1k_1}}(B_{1k_1,\epsilon})\ \E \prod_{k_2} Q_2^{r_{2k_2}}(B_{2k_2,\epsilon}).
\end{gather*}
Returning to \eqref{eq:camerlengthi}, we have that
\begin{gather*}
    M_{\n}\left( \bigtimes_{k_1,k_2} B_{k_1,k_2,\epsilon} \right)= \\
    \frac{1}{\Gamma(n)} \sum_{r \in \mathcal R} M_{r_1} \left( \bigtimes_{k_1} B_{1k_1, \epsilon} \right) M_{r_2} \left( \bigtimes_{k_2} B_{2k_2,\epsilon} \right)  \rho^{r} \int_{0}^{\infty} u^{n-1}e^{-\rho \Psi(u)} \lb \prod_{k_1,k_2} \xi_{n_{k_1k_2}, r_{k_1k_2}}(u) \rb \diff u,
\end{gather*}
where $M_{ r_v}(d \theta_{v1}, \dots, \theta_{vK_v}) = \E \prod_{k_v} Q_v^{r_{vk_v}}(d \theta_{vk_v})$ for $v=1,2$. Using Proposition 3 of \cite{james2009posterior}, we have that for $v=1,2$,
\begin{equation*}
    M_{ r_v}(d \theta_{v1}, \dots, \theta_{vK_v}) = \left( \prod_{k_v} H_v(d \theta_{vk_v}) \right) \frac{\gamma_v^{K_v}}{(\gamma_v)^{(| r_v|)}} \prod_{k_1} (r_{vk_v}-1)!,
\end{equation*}
where $|r_v| = \sum_{k_v=1}^{K_v} r_{vk_v} = r$ for $v=1,2$. The MEPPF is obtained by taking $\epsilon \to 0$, then applying Example 1 in \cite{camerlenghi2018bayesian}.
\end{proof}

\subsection{Proof of Theorem \ref{thm:ERI}}
\begin{proof}
    We express the PCDP in terms of a stick-breaking representation using the technique in \cite{teh2006sharing} for the HDP stick-breaking decomposition. We first begin by sampling the marginals $Q_v$. We write $Q_v = \sum_{k_v=1}^{\infty} q_{vk_v} \delta_{\theta^*_{vk_v}}$, where $\theta^*_{vk_v} \sim H_v$, $q_{vk_v} = \beta_{vk_v} \prod_{m_v<k_v}(1-\beta_{vm_v})$, and $\beta_{vk_v} \sim \tx{Beta}(1, \gamma_v)$. We collect the individual weights of each $Q_v$ into the infinite-dimensional probability vectors $q_v = (q_{vk_v})_{k_v=1}^\infty$. Observe that the support of $P$ is comprised of the Cartesian product of the supports of $Q_1, Q_2$. This implies that we can rewrite $P$ as
    \begin{equation} \label{eq:CLIC-DP-stick}
        P = \sum_{k_1=1}^{\infty} \sum_{k_2=1}^{\infty} p_{k_1 k_2} \delta_{(\theta_{1k_1}^*, \theta_{2k_2}^*)},
    \end{equation}
    where $0 \leq p_{k_1 k_2} \leq 1$ and $\sum_{k_1=1}^{\infty} \sum_{k_2=1}^{\infty} p_{k_1 k_2} = 1$. We denote $p = (p_{k_1k_2})_{k_1, k_2}$ to be the matrix of probabilities constructed from the weights in \eqref{eq:CLIC-DP-stick}. Assuming that $H_1, H_2$ are non-atomic, the probability matrix is itself a Dirichlet process centered on the product measure of $q_1$ and $q_2$,
    \begin{align*}
        p  \mid \rho, q & \sim \DP(\rho, q ); & q = q_1 \otimes q_2,
    \end{align*}
    where $\otimes$ is the tensor product.    
    Conditional on $p$, 
    \begin{equation} \label{eq:rand-proof-infinite-sum}
        \E[R(C_1,C_2) \mid  p] = \sum_{k_1=1}^\infty \sum_{k_2=1}^\infty p_{k_1k_2}^2 + \sum_{k_1 \neq m_1}^\infty \sum_{k_2 \neq m_2}^\infty p_{k_1 k_2} p_{m_1 m_2}.
    \end{equation}
    Observe that $\E[R(C_1, C_2) \mid p]$ is an infinite sum. However, all terms in the sum are positive and the Rand index is bounded above by $1$, so we can exploit the linearity of expectation by invoking Fubini's theorem, which we will implicitly invoke multiple times by using the law of total expectation. By definition, $p_{k_1 k_2} \mid \rho, q \sim \tx{Beta}(\rho q_{1 k_1} q_{2 k_2}, \rho(1 - q_{1 k_1} q_{2 k_2}))$. This implies
    \begin{gather*}
        \E[p_{k_1 k_2}^2 \mid q_{1k_1}, q_{2k_2}] = \frac{q_{1k_1} q_{2k_2}(1 - q_{1k_1} q_{2k_2})}{\rho + 1} + q_{1k_1}^2 q_{2k_2}^2; \\
        \E[q_{k_1 k_2} q_{m_1 m_2} \mid q_{1 k_1}, q_{1 m_1}, q_{2 k_2}, q_{2 m_2}] = \left( 1 - \frac{1}{\rho + 1} \right) q_{1k_1} q_{1 m_1} q_{2 k_2} q_{2 m_2}. 
    \end{gather*}
    Similarly, we have that $q_1 \ind q_2$ and distributed as stick-breaking processes with concentration parameters $\gamma$. Calculating the expectation of \eqref{eq:rand-proof-infinite-sum} requires computing the moments of a stick-breaking process. For $k_v > m_v$, these are
    \begin{gather*}
        \E[q_{v k_v}] = \frac{\gamma^{k_v-1}}{(1+\gamma)^{k_v}} \\ 
        \E[q_{v k_v}^2] = \lb  \frac{\gamma}{(\gamma+1)^2(\gamma+2)} + \frac{1}{(1 + \gamma )^2}  \rb \lb \frac{\gamma}{(\gamma+1)^2(\gamma+2)} + \frac{\gamma^2}{(1 + \gamma )^2} \rb^{k_v-1} \\ 
        \E[q_{vk_v} q_{v m_v}] = \frac{\gamma}{(\gamma + 1)^2(\gamma+2)} \left( \frac{\gamma}{\gamma+1} \right)^{k_v - m_v - 1} \lb \frac{\gamma}{(\gamma+1)^2(\gamma+2)} + \frac{\gamma^2}{(1 + \gamma )^2}  \rb^{m_v-1}
     \end{gather*}
     By plugging these terms back into \eqref{eq:rand-proof-infinite-sum}, we then need to compute the limits of several power series, all of which are geometric. In particular, we have that
     \begin{equation}
         \begin{gathered} \label{eq:series-sums}
        \sum_{k_v=1}^{\infty} \E[q_{v k_v}] = 1;\quad 
         \sum_{k_v=1}^{\infty} \E[q_{v k_v}^2] = \frac{1}{\gamma + 1};\quad 
         2 \sum_{m_v=1}^{\infty} \sum_{k_v = m_v + 1}^\infty \E[q_{v k_v} q_{v m_v}] = \frac{\gamma}{\gamma + 1}.
         \end{gathered}
     \end{equation}
    Plugging these identities back into \eqref{eq:rand-proof-infinite-sum} shows that
     \begin{gather*}
         \tau_{12} = \frac{1}{\rho + 1} + \left( 1 - \frac{1}{\rho+1} \right) \left( \frac{1}{(\gamma+1)^2} + \frac{\gamma^2}{(\gamma+1)^2}  \right)  \\
         = \frac{1}{\rho + 1} + \left( 1 - \frac{1}{\rho+1} \right)  \frac{\gamma^2 + 1}{(\gamma+1)^2}.
     \end{gather*}
     By Corollary 1 in \cite{page2022dependent}, the expected Rand index of two independent $\tx{CRP}(\gamma)$ partitions is $(\gamma^2 + 1)/(\gamma+1)^2$. This shows how $\rho$ shrinks toward independent $\tx{CRP}(\gamma)$ partitions. Additionally, increased $\gamma$ pulls the expected Rand index to one, reflecting the tendency of $(C_1, C_2)$ to concentrate towards a $\tx{CRP}(\rho)$ distribution for large $\gamma$.

     Observe that in the case where $\gamma_1 \neq \gamma_2$, \eqref{eq:series-sums} implies that
     \begin{gather*}
         \tau_{12} = \frac{1}{\rho+1} + \left( 1 - \frac{1}{\rho+1} \right) \lb \frac{1}{(\gamma_1+1)(\gamma_2+1)} + \frac{\gamma_1\gamma_2}{(\gamma_1+1)(\gamma_2 + 1)}  \rb \\
         = \frac{1}{\rho+1} + \left( 1 - \frac{1}{\rho+1} \right) \frac{1+\gamma_1 \gamma_2}{(\gamma_1+1)(\gamma_2+1)}.
     \end{gather*}
\end{proof}

\subsection{ERI of the Finite Approximation}
We now return to the finite approximation to the PCDP presented in Section \ref{section:PCDP}, which we use during our simulations and application. As stated in the main article, the ERI under the finite approximation converges to the ERI of the PCDP as the number of components is taken to infinity.

\begin{prop}
    Let $\tilde{\tau}_{12}$ be the ERI of $C_1$ and $C_2$ under the finite approximation to the PCDP in \eqref{eq:finite-approximation}. Then $\tilde{\tau}_{12} = \nu + (1 - \nu) \kappa(L_1,L_2)$, where $\kappa(L_1, L_2) \to (1+\gamma^2)/(1+\gamma)^2$ as $L_1, L_2 \to \infty$.
\end{prop}

\begin{proof}
    Similar to the infinite case, for any probability matrix $\tilde p$, one can show that
    \begin{equation*}
        \E[R(C_1, C_2) \mid \tilde p] = \sum_{k_1=1}^{L_1} \sum_{k_2 = 1}^{L_2} \tilde p_{k_1 k_2}^2 + \sum_{k_1 \neq m_1} \sum_{k_2 \neq m_2} \tilde p_{k_1k_2} \tilde p_{m_1 m_2}.
    \end{equation*}
    If $\tilde p = \tilde p_1 \otimes \tilde p_2$, with $ \tilde p_1 \sim \tx{Dir}(\gamma/L_1, \dots, \gamma/L_1)$, $ \tilde p_2 \sim \tx{Dir}(\gamma/L_2, \dots, \gamma/L_2)$, and $ \tilde p_1 \ind \tilde p_2$, then we have that
    \begin{gather*}
        \E[\tilde p_{k_1 k_2}^2 ] = \E[ \tilde p_{1k_1}^2] \E [ \tilde p^2_{2k_2}] = \lb  \frac{(1/L_1)(1-1/L_1)}{\gamma + 1} + \frac{1}{L_1^2} \rb \lb  \frac{(1/L_2)(1-1/L_2)}{\gamma + 1} + \frac{1}{L_2^{2}} \rb \\
        \E[\tilde p_{k_1k_2} \tilde p_{m_1m_2}]  = \E[\tilde p_{1k_1} \tilde p_{1m_1}] \E[\tilde p_{2k_2} \tilde p_{2m_2}] = \frac{\gamma^2}{(\gamma + 1)^2} \frac{1}{L_1^2 L_2^{ 2}}.
    \end{gather*}
    So for $C_1 \ind C_2$, with both distributed as symmetric Dirichlet-multinomial partitions,
    \begin{gather*}
        \E[R(C_1,C_2) \mid \gamma] = \kappa(L_1,L_2) \nonumber \\
        = \left(  \frac{1-1/L_1}{\gamma + 1} + \frac{1}{L_1^2}    \right) \left( \frac{1 - 1/L_2}{\gamma + 1} + \frac{1}{L_2^2}  \right) + \frac{(1-1/L_1)(1-1/L_2) \gamma^2}{(\gamma+1)^2}.
    \end{gather*}
    Note that $\lim_{L_1,L_2 \to \infty} \kappa(L_1,L_2) = (1+\gamma^2)/(1+\gamma)^2$, the ERI for two independent $\CRP(\gamma)$ partitions.

    Next, assume $(C_1,C_2)$ is distributed as partitions induced by the finite approximation to the PCDP in \eqref{eq:finite-approximation} with $\gamma_1 = \gamma_2 = \gamma$. Conditional on $\tilde q_1, \tilde q_2$,
    \begin{gather*}
        \E[\tilde p_{k_1k_2}^2 \mid \tilde q_{1k_1}, \tilde q_{2k_2}] = \frac{ \tilde q_{1k_1} \tilde q_{2k_2} (1 - \tilde q_{1k_1} \tilde q_{2k_2})   }{ \rho + 1} + \tilde q_{1k_1}^2 \tilde q_{2k_2}^2 \\
        \E[\tilde p_{k_1k_2} \tilde p_{m_1m_2} \mid \tilde q_{1k_1}, \tilde q_{2k_2} ] = \left(1 - \frac{1}{\rho+1} \right) \tilde q_{1k_1} \tilde q_{2k_2} \tilde q_{1m_1} \tilde q_{2m_2}.  
    \end{gather*}
    Set $\nu = 1/(\rho + 1)$. Recall that $\tilde q_v \sim \tx{Dir}(\gamma/L_v, \dots, \gamma/L_v)$ for $v=1,2$. Hence,
    \begin{gather}
        \E[\tilde p_{k_1k_2}^2] = \frac{\nu}{L_1 L_2} + (1-\nu) \lb \frac{(1/L_1)(1-1/L_1)}{\gamma + 1} + \frac{1}{L_1^2} \rb \lb \frac{(1/L_2)(1-1/L_2)}{\gamma + 1} + \frac{1}{L_2^2} \rb; \label{eq:akk-second-moment} \\ 
        \E[\tilde p_{k_1k_2} \tilde p_{m_1m_2}] = (1-\nu) \frac{ \gamma^2}{L_1^1 L_2^2(\gamma+1)^2} . \label{eq:akk-covariance}
    \end{gather}
    The result follows by multiplying \eqref{eq:akk-second-moment} and \eqref{eq:akk-covariance} by $L_1 L_2$ and $L_1L_2(L_1-1)(L_2-1)$, respectively, then summing them together.
\end{proof}

\begin{remark}
    Let $L_1 = L_2 = L$. Then $|\tau_{12} - \tilde \tau_{12}| \propto \lb (\gamma/L) - \gamma + 1 \rb/L$, where proportionality is in terms of $L$. For example, if $\gamma=1$, then $|\tau_{12} - \tilde \tau_{12}| \propto 1/L^2$.
\end{remark}
\begin{proof}
    Using equations \eqref{eq:akk-second-moment} and \eqref{eq:akk-covariance}, we can see that
    \begin{gather*}
        \tilde \tau_{12} - \tau_{12} \\
        = (1-\nu) \lb \left( \frac{1 - 1/L}{\gamma + 1} + 1/L  \right)^2 - \frac{1}{(\gamma+1)^2} + \left( (1-1/L)^2 - 1 \right) \frac{\gamma^2}{(\gamma+1)^2}  \rb \\
        = (1 - \nu) \lb \left( 1/L^2 - 2/L \right) \frac{\gamma^2 + 1}{(\gamma+1)^2} + \frac{(2/L)(1-1/L)}{\gamma+1} + 1/L^2 \rb \\
        = \frac{2 \gamma(1-\nu)}{L(\gamma+1)^2} \lb \frac{\gamma}{L} - \gamma + 1 \rb.
    \end{gather*}
\end{proof}

\section{Additional Details on Simulations}
Code for all simulations and the application to the CPP data is available in the \href{https://github.com/adombowsky/clic}{CLIC Github repository}. For the simulations in Section \ref{section:computation}, we set $\theta_0 = 0$ and $\sigma_0 = 1$. For the priors on $\sigma_v^2$, we use the $\tx{Gamma}(1,1)$ distribution, and for the number of components we set $L_v = 5$. The t-HDP of \cite{franzolini2023bayesian} is implemented using code from the Github page of the first author, where the base measures are taken to be $\mathcal{N-IG}(0,1,1,1)$, the DP concentration parameters are both set equal to $1$, and the number of components are both set to be $5$. The EM algorithm is implemented with the \texttt{R} package \texttt{mclust} \citep{scrucca2023mclust}. The clustering point estimates are computed with the \texttt{R} package \texttt{mcclust} \citep{fritsch2009package} and \texttt{mcclust.ext} \citep{wade2015package} using the \texttt{min.VI()} function, which we also use in the application to the \cite{longnecker2001association} data in Section \ref{section:application}.

\subsection{Joint and Conditional Samplers for the Labels}
A key step in the Gibbs sampler is imputing $(c_{1i}, c_{2i})$ for each object. There are multiple strategies for sampling the labels a priori given the clustering formations $C_v^{-i}$, where $C_v^{-i}$ is the $v$th partition with object $i$ removed. In the \textit{joint sampler}, we sample $(c_{1i}, c_{2i})$ using a matrix multinomial given by
\begin{equation} \label{eq:joint-sampler}
    \pi(c_{1i} = k_1, c_{2i} = k_2 \mid C_1^{-i}, C_2^{-i}, \rho) \propto \frac{\rho \gamma_1 \gamma_2 }{L_1L_2} + n_{k_1k_2}^{-i}.
\end{equation}
Alternatively, the conditional sampler decomposes the joint distribution of $c_{1i}$ and $c_{2i}$ as
\begin{gather}
    \pi(c_{1i} = k_1, c_{2i} = k_2 \mid C_1^{-i}, C_2^{-i}, \rho) \nonumber \\
    = \pi(c_{2i} = k_2 \mid c_{1i} = k_1, C_1^{-i}, C_2^{-i}, \rho) \pi(c_{1i} = k_1 \mid  C_1^{-i}, C_2^{-i}, \rho).
\end{gather}
By summing over $k_2$ in \eqref{eq:joint-sampler}, we see that
\begin{gather}
    \pi(c_{1i} = k_1 \mid C_1^{-i}, C_2^{-i}, \rho) \propto \frac{\rho \gamma_1 \gamma_2}{L_1} + n_{1k_1}^{-i}, \label{eq:C1-marginal} \\
    \pi(c_{2i}= k_2 \mid c_{1i} = k_1, C_1^{-i}, C_2^{-i}, \rho) \propto \frac{\rho \gamma_1 \gamma_2}{L_1L_2} + n_{k_1k_2}^{-i}. \label{eq:C2-conditional}
\end{gather}
Note that the conditional distribution of $c_{2i}$ in \eqref{eq:C2-conditional} only depends on objects with the same value of $c_{1i}$. When $\rho$ is fixed, we have that
\begin{equation}
    \pi(c_{1i} = k_1, c_{2i} = k_2 \mid -) \propto f(X_i; (\theta_{1k_1}^*, \theta_{2k_2}^*)) \lb \frac{\rho \gamma_1 \gamma_2}{L_1} + n_{1k_1}^{-i} \rb \lb \frac{\rho \gamma_1 \gamma_2}{L_1L_2} + n_{k_1k_2}^{-i} \rb. 
\end{equation}
Since $f(X_i; (\theta_{1k_1}^*, \theta_{2k_2}^*))$ factorizes into a product, we can use compositional sampling to simulate $\pi(c_{1i} \mid -)$ and then $\pi(c_{2i} \mid c_{1i}, -)$. For the random $\rho$ case, we will need to include the latent probability vectors $\tilde q_1$ and $\tilde q_2$. In that case, we have that
\begin{equation}
    \pi(c_{1i} = k_1, c_{2i} = k_2 \mid C_1^{-i}, C_2^{-i}, \tilde q_1, \tilde q_2, \rho) \propto \rho \tilde q_{1k_1} \tilde{q}_{2k_2} + n_{k_1k_2}^{-i},
\end{equation}
and so
\begin{equation}
    \pi(c_{1i} = k_1, c_{2i} = k_2 \mid -) \propto f(X_i; (\theta_{1k_1}^*, \theta_{2k_2}^*)) \lb \rho \tilde q_{1k_1} + n_{1k_1}^{-i} \rb \lb \rho \tilde q_{1k_1} \tilde q_{2k_2} + n_{k_1k_2}^{-i} \rb. 
\end{equation}

In general, the conditional sampling scheme leads to far better mixing of the partitions than the joint sampling scheme. Observe that the reverse conditional, i.e. doing compositional sampling of $c_{2i}$ then $c_{1i} \mid c_{2i}$, is completely symmetric. Therefore, one can do either conditional sampler in practice, or even alternate between the conditional samplers from iteration to iteration. We have found that either order/alternation leads to similar efficiency, so for our simulations and the application, we implement the scheme outlined above. 

When $V \geq 2$, we construct a similar sampler by taking sums across the dimensions of the $V$-dimensional contingency array. For the uncorrelated views model, the full conditional of the labels for object $i$ is
\begin{equation} \label{eq:V-partitions-full-conditional}
    \pi(c_{1i}=k_1, \dots, c_{Vi} = k_{V} \mid -) \propto \prod_{v=1}^V f_v(X_{vi}; \theta_{vk_v}^*) \times \lb \rho \prod_{v=1}^V \tilde q_{vk_v} + n_{k_1 \cdots k_V}^{-i} \rb,
\end{equation}
where $n_{k_1 \cdots k_V}^{-i}$ are entries in the contingency array for $C_1^{-i}, \dots, C_V^{-i}$. One conditional sampling scheme is
\begin{gather} \label{eq:V-partitions-Markov}
    \pi(c_{1i}=k_1, \dots, c_{Vi} = k_{V} \mid -) \nonumber \\
    \propto f_1(X_{1i}; \theta_{1i}^*) \lb \rho \tilde q_{1k_1} + n_{1k_1}^{-i} \rb \prod_{v=2}^V f_v(X_{vi}; \theta_{vk_v}^*) \lb \rho \prod_{u=1}^v \tilde q_{uk_u} + n_{k_1 \cdots k_u}^{-i} \rb,
\end{gather}
but analogous equations hold for any permutation of $[V]$, e.g., if $\omega$ is a permutation of $[V]$, then
\begin{gather} \label{eq:V-partitions-Markov-permutation}
    \pi(c_{\omega(1)i}=k_{\omega(1)}, \dots, c_{\omega(V)i} = k_{\omega(V)} \mid -) \nonumber \\
    \propto f_{\omega(1)}(X_{\omega(1)i}; \theta_{\omega(1)i}^*) \lb \rho \tilde q_{\omega(1)k_{\omega(1)}} + n_{\omega(1)k_{\omega(1)}}^{-i} \rb  \\
    \times \prod_{v=2}^V f_{\omega(v)}(X_{\omega(v)i}; \theta_{\omega(v)k_v}^*) \lb \rho \prod_{u=1}^v \tilde q_{\omega(u)k_{\omega(u)}} + n_{k_{\omega(1)} \cdots k_{\omega(u)}}^{-i} \rb.
\end{gather}
A similar guarantee holds for the correlated views model.

\subsection{Additional Remarks on Two Uncorrelated Views Illustration}

\begin{figure}
    \centering
    \includegraphics[scale=0.6]{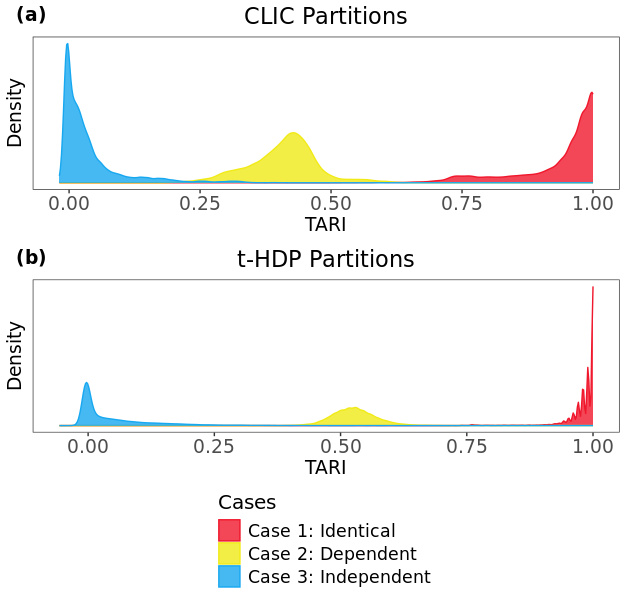}
    \caption{Posterior distributions of the TARI for CLIC and the t-HDP under the first overlap ($\eta^2 = 0.2$).}
    \label{fig:rand-multiview-1}
\end{figure}

\begin{figure}
    \centering
    \includegraphics[scale=0.6]{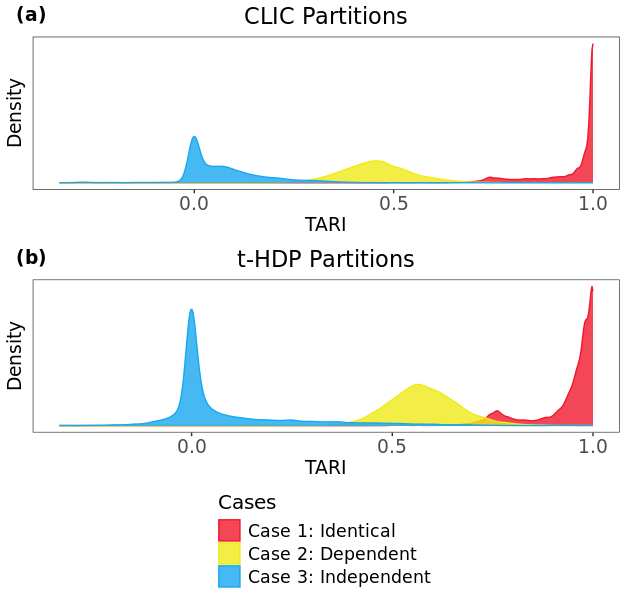}
    \caption{Posterior distributions of the TARI for CLIC and the t-HDP under the second overlap ($\eta^2 = 0.45$).}
    \label{fig:rand-multiview-2}
\end{figure}

The dependence between partitions aids inference on $C_2$, allowing CLIC to outperform competitors. In addition, samples from $\pi(C_1, C_2 \mid \X)$ make it possible to infer the degree of dependence between the partitions. Figures \ref{fig:rand-multiview-1} and \ref{fig:rand-multiview-2} show the estimated posterior distributions of the TARI for CLIC and the t-HDP under the two different overlaps. In both overlaps, CLIC can model a wide range of dependence structures for random partitions. In cases 1 and 3, the TARI posterior concentrates near 1 and 0, respectively, whereas for case 2, the posterior concentrates between 0.25 and 0.75. We can see similar, but distinct posterior distributions for the t-HDP. For instance, in case 2 under both overlaps, the t-HDP tends to produce clusterings with a higher TARI than those from the CLIC posterior. Hence, while both strategies may lead to similar point estimates, inference on the clustering dependence structure differs between CLIC and the t-HDP. In addition, the computation time for each method is recorded in Table \ref{table:two-view-multiview-time}.

\begin{table}[ht]
    \centering
    \begin{tabular}{cccccc}
    \toprule
       Overlap & Case & CLIC & t-HDP & EM & IDPs \\
      \multirow{3}{4em}{$\eta^2 \hspace{-0.5mm} = \hspace{-0.5mm} 0.2$} & 1 &  31.150 & 314.780 & 0.023 & 28.177 \\
       & 2 &  30.105 &   314.860 & 0.019 & 15.681 \\
       & 3 & 30.161  & 316.797 & 0.022 & 27.405 \\
       \multirow{3}{4em}{$\eta^2  \hspace{-0.5mm} = \hspace{-0.5mm} 0.45$} & 1 &   32.512 & 316.287 &   0.025 & 29.427 \\
       & 2 & 30.903 & 314.624 &  0.034 & 15.843 \\
       & 3 & 29.372 & 316.704 &  0.022  & 27.573 \\
        \bottomrule
    \end{tabular}
    \caption{Computation time (in seconds) for CLIC and competitors in the two uncorrelated view simulation study, including time to run the MCMC sampler and calculate the point estimate. All Bayesian methods are run for 30,000 iterations. CLIC is substantially more efficient than the t-HDP, though the EM algorithm is the fastest option.}
    \label{table:two-view-multiview-time}
\end{table}

\subsection{Illustration with Two Correlated Views} \label{section:conditional-simulations}

We simulate the data with $(X_i \mid c_{1i}, c_{2i}) \sim \N(X_{1i}; \lb -1\rb^{c_{1i}+1}, 0.2) \times \N(X_{2i};\lb -1 \rb^{c_{2i}} X_{1i}, \eta^2)$, where $\eta^2$ takes the same values as in the uncorrelated setting in Section \ref{section:simulations}. If $c_{2i}=1$ observation $i$ has a negative relationship with the covariates, while if $c_{2i}=2$ they have a positive association. Figure \ref{fig:conditional-data} shows the true values for $C_1$ and $C_2$ along with the synthetic data. For the PCDP, we assume $\pi(X_i; \theta_i) = \N(X_{1i}; \theta_{1i}, \sigma_1^2) \times \N(X_{2i}; \theta_{2i} X_{1i}, \sigma_2^2)$ and $H_v(\theta) = \N(\theta; \mu_0, \sigma_0^2)$. We compare CLIC with a correlated implementation of the t-HDP (which we call the ct-HDP), two IDPs where one DP clusters $\X_1$ and the other clusters $\X_2 \mid \X_1$ using a mixture of linear regressions, and the EM algorithm applied to $\X_1$ and a linear regression implementation of the EM algorithm \citep{leisch2004flexmix} applied to $\X_2 \mid \X_1$. For both CLIC and the IDPs, the point estimates are obtained by minimizing the VI loss. 

\begin{figure}
    \centering
    \includegraphics[scale=0.58]{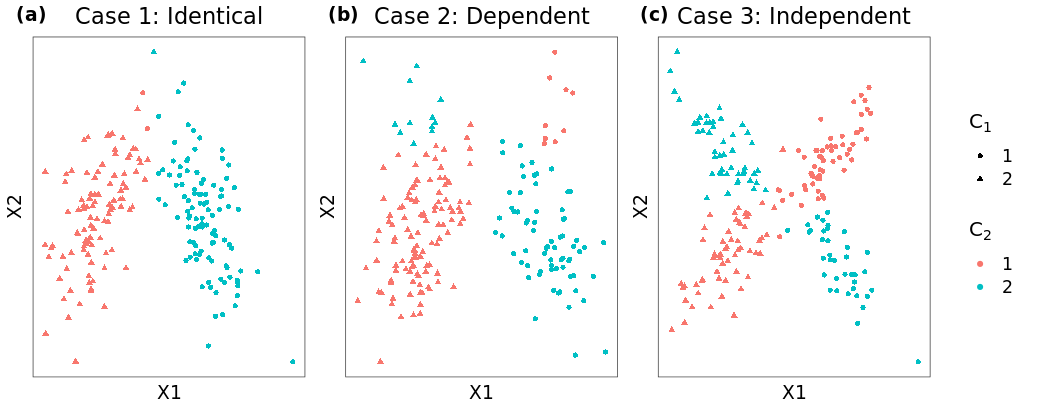}
    \caption{The synthetic data simulated in the correlated views setting, where shapes and colors correspond to the true values of $C_1$ and $C_2$. Here, we have that $C_2$ clusters are characterized by differing regression slopes.}
    \label{fig:conditional-data}
\end{figure}

\begin{table}[ht]
    \centering
    \begin{tabular}{cccccc}
    \toprule
         & CLIC & ct-HDP &  IDPs & EM \\
         \midrule
        Case 1 &  (0.941, 0.941) & (0.960, 0.960) & (0.941, 0.531) & (0.941, 0.606) \\
        Case 2 & (1.000, 0.775) & (1.000, 0.775) & (1.000, 0.704) &   (1.000, 0.721) \\
        Case 3 & (0.921, 0.883)  & (0.921, 0.883) & (0.921, 0.883) & (0.921, 0.883) \\
        \bottomrule
    \end{tabular}
    \caption{The adjusted Rand indices (ARI) between the true $(C_1, C_2)$ and the point estimates computed by CLIC, the ct-HDP, two IDPs (where the second model uses profile regression of $\X_2$ on $\X_1$), and independent EM for the synthetic data. CLIC has the joint highest accuracy for estimating $(C_1,C_2)$ with the exception of case 1, where the ct-HDP correctly classifies one extra object than CLIC.}
    \label{table:aris-correlated}
\end{table}

The corresponding ARIs are given in Table \ref{table:aris-correlated}. CLIC coherently clusters the two views and attains high ARI values in all three cases. In case 1, the outcome can be predicted by first partitioning the covariates and then fitting a regression model within each cluster. CLIC and the ct-HDP are able to identify and exploit this relationship, leading to ARIs of 0.941 and 0.960, respectively, greatly outperforming IDPs and EM. This task becomes far more difficult in case 2, where objects in either of the true clusters in $C_1$ can have positive or negative correlation between $X_{1i}$ and $X_{2i}$. However, objects in cluster 1 of $C_1$ are more likely to have a positive association, and objects in cluster 2 of $C_1$ tend to have a negative relationship. All methods find this regime challenging, in part because the separation between the $C_1$ clusters leads to wide gaps in the regression lines. However, CLIC and the ct-HDP still outperform IDPs and EM in classifying the objects because they account for clustering dependence. In case 3, objects in either cluster of $C_1$ are equally likely to have positive or negative correlation.  Due to the separation between the clusters, all four methods return the same point estimate and achieve high ARI values with the true partitions. A comparison between the posterior distributions of the TARI under CLIC and the t-HDP is given in Figure \ref{fig:rand-plots-conditional}. As in the uncorrelated views setting, the posterior for the TARI in case 1 is more diffuse for the t-HDP, and the posterior for case 3 is more concentrated at $0$ for CLIC. 

\begin{figure}[ht]
    \centering
    \includegraphics[scale=0.75]{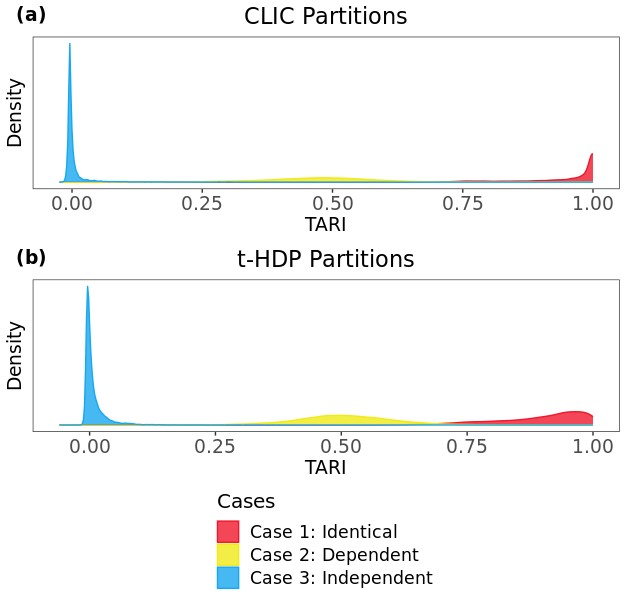}
    \caption{Estimated posterior distributions of the TARI under two correlated views for CLIC and the t-HDP based on synthetic data.}
    \label{fig:rand-plots-conditional}
\end{figure}

\subsection{Additional Remarks on Three Views Setting}
The computation times for CLIC, the t-HDP, EM, and the IDPs were $60.029$, $518.545$, $0.033$, and $23.477$ seconds respectively, where each Bayesian algorithm is run for $30,000$ iterations.

\subsection{Additional Remarks on Varying Dimension and Sample Size Simulations}

All Bayesian methods are ran for $30,000$ iterations with $5,000$ as warm-up, and every other iteration was thinned. Computation times for all methods are displayed in Table \ref{table:varying-dim-samp-size-time}. The PCDP uses base measures $H_v(\theta) = \N_{d_v}(\theta; \bs 0_{d_v}, I_{d_v})$, where $0_{d_v}$ and $I_{d_v}$ are the $d_v$-dimensional zero vector and identity matrix. We also set $\gamma_1 = \gamma_2 = 1$, $L_1 = L_2 = 10$, and  $1/\sigma_{v}^2 \sim \tx{Gamma}(1,1)$. For the t-HDP, we use the function \texttt{telescopic\_HDP\_NNIW\_multi.R} from the \href{https://github.com/beatricefranzolini/CPE/blob/main/telescopic_HDP_NNIW_multi.R}{Github repository} for the method. We use the default hyperparameters $H_0 = H = 10$, concentration parameters equal to $0.1$, and Normal-inverse-Wishart scale equal to $0.1$. Like the PCDP, the IDPs and DP also have standard Gaussian base measures, $10$ components, and a variance prior set to a $\tx{Gamma}(1,1)$ distribution. The EM algorithm is implemented in the same manner as the other simulations.

\begin{table}[]
    \centering
    \begin{tabular}{cccccc}
    \toprule
     \multicolumn{6}{c}{$n=100$} \\
     \midrule
         $d_2$ & CLIC & t-HDP & EM & IDPs & DP \\
         2 & 16.960 & 11875.856 & 0.960 & 11.106 & 6.343 \\
         10 & 18.685 & 14577.428 & 0.926 & 12.774  & 8.040 \\
         25 & 22.918 & 16182.596 & 1.235 &   16.564 &  11.883 \\
         \midrule
    \end{tabular}
    \centering
    \begin{tabular}{cccccc}
    \toprule
     \multicolumn{6}{c}{$n=200$} \\
     \midrule
         $d_2$ & CLIC & t-HDP & EM & IDPs & DP \\
         2 & 62.835 & 17647.145 & 3.139    & 29.664 &  16.702 \\
         10 & 66.319 & 21457.653 &  3.004  &   33.316  &  19.830 \\
         25 & 73.641 & 23917.661 & 2.041  &  40.047  &  26.918 \\
         \midrule
    \end{tabular}
    \centering
    \begin{tabular}{cccccc}
    \toprule
     \multicolumn{6}{c}{$n=500$} \\
     \midrule
         $d_2$ & CLIC & t-HDP & EM & IDPs & DP \\
         2 & 294.488 & 34363.581 & 10.147 & 136.790 & 81.237 \\
         10 & 300.070 & 41564.017 & 11.273   & 147.140  &  80.839 \\
         25 & 317.207 & 46181.697 & 39.278   & 161.755 & 96.035 \\
         \bottomrule
    \end{tabular}
    \caption{Computation time (in seconds) for all methods to compute the point estimates along both views.}
    \label{table:varying-dim-samp-size-time}
\end{table}

\subsection{Mixing and Convergence for the Two Views Setting}

\begin{figure}[ht]
    \centering
    \includegraphics[scale=0.5]{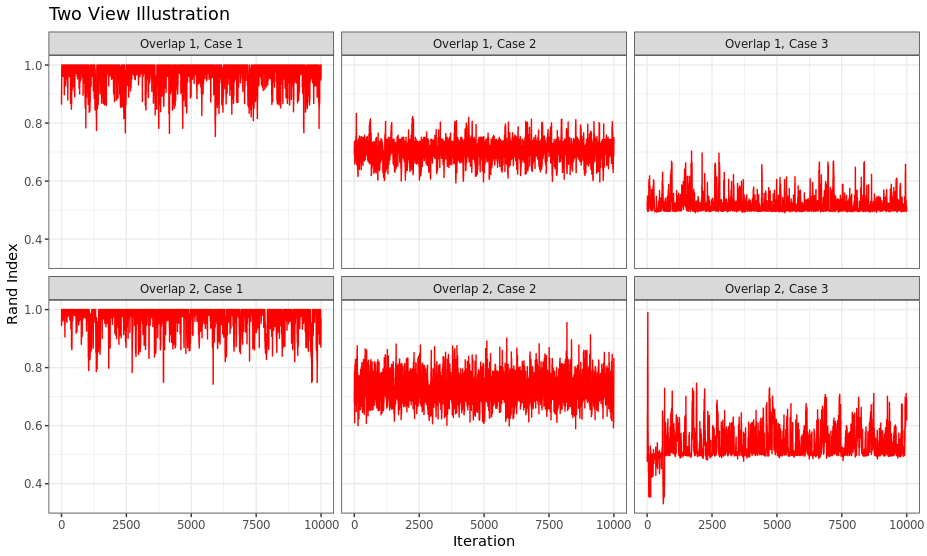}
    \caption{Traceplots for $R(C_1,C_2)$ for each case and overlap (i.e., value of $\eta^2$) in the two view setting. The Rand index may mix worse when it hits either its upper or lower boundary, and mixes better when it is in $(0,1)$.}
    \label{fig:rand-trace-plots}
\end{figure}

\begin{figure}
    \centering
    \includegraphics[scale=0.65]{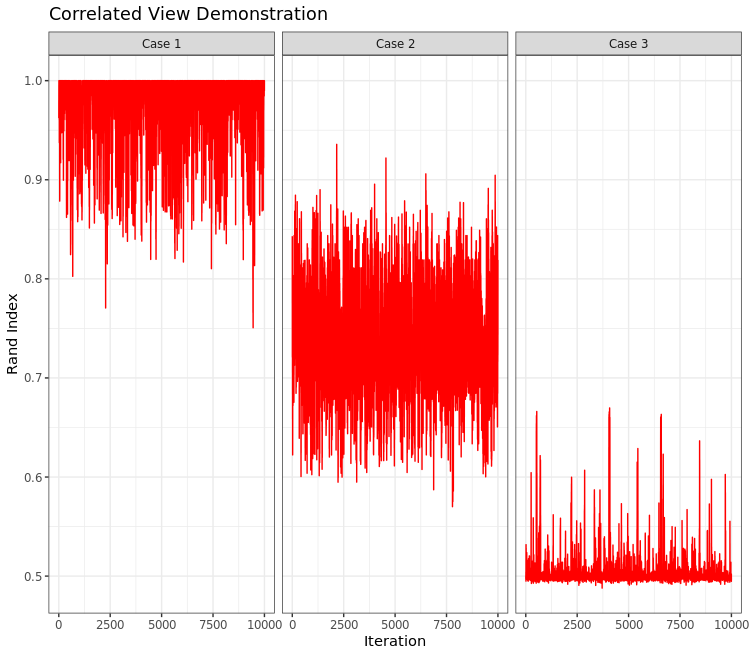}
    \caption{Traceplots for $R(C_1,C_2)$ in the illustration discussed in Section \ref{section:conditional-simulations}.}
    \label{fig:rand-conditional-trace-plots}
\end{figure}

\begin{table}[ht]
    \centering
    \begin{tabular}{cccc}
    \toprule
        Overlap & Case 1 & Case 2 & Case 3 \\
         \midrule
        $\eta^2=0.2$ & 480.852 & 598.751 & 378.956 \\
        $\eta^2 = 0.45$ & 366.986 & 919.625 & 219.241  \\
        \bottomrule
    \end{tabular}
    \caption{Effective sample size of $R(C_1,C_2)$ for each case and overlap in the two uncorrelated views setting.}
    \label{table:ess}
\end{table}

We find that the mixing for the partitions is much better when $\rho$ is given a hyperprior and updated during the MCMC sampler than when it is fixed. In the main article, we present two alternative priors and corresponding sampling strategies for updating $\rho.$ The first is to use a gamma prior as in \cite{escobar1995bayesian}, which is semiconjugate. The second is to use a griddy-Gibbs sampler \citep{ritter1992facilitating} and allow $\rho$ to vary across a pre-specified grid. In the simulations, we assume that $\rho \sim \tx{Gamma}(1,1)$ and fix $\gamma_1 = \gamma_2 = \gamma = 1$.

Recall that $0 \leq R(C_1, C_2) \leq 1$. Trace plots for $R(C_1,C_2)$ are given in Figures \ref{fig:rand-trace-plots} (for uncorrelated views) and \ref{fig:rand-conditional-trace-plots} (for correlated views). Additionally, the effective sample sizes for the uncorrelated views simulations are in Table \ref{table:ess}, and for the correlated views simulations these values are $558.467$ (case 1), $678.036$ (case 2), and $335.535$ (case 3). Convergence of the Rand index in all simulations is rapid, but the Gibbs sampler is most efficient when the true partitions are weakly dependent (i.e., case 2). Efficiency can decrease when the true partitions exhibit either case 1 or case 3. In case 1, the true Rand index is equal to 1, and so samples of $R(C_1,C_2)$ tend to accumulate near its upper bound, which results in worse mixing than the weak dependence case. While $R(C_1,C_2)$ is bounded below by $0$, we show in Theorem \ref{thm:ERI} that $\E[R(C_1,C_2)] = \tau_{12} \geq (1+\gamma^2)/(1+\gamma)^2 = 1/2$. Though there is still positive posterior probability that $R(C_1, C_2) < 1/2$, the vast majority of samples from the Rand index are at least as large as $1/2$. Hence, in case 3, samples of $R(C_1,C_2)$ accumulate near an (approximate) lower bound, resulting in correlated draws.

\subsection{Mixing and Convergence for the Three Views Setting}
Traceplots for the pairwise Rand indices are given in Figure \ref{fig:three-view-traces}. The effective sample size is $465.9618$ for $R(C_1, C_2)$, $565.8382$ for $R(C_1,C_3)$, and $365.5251$ for $R(C_2,C_3)$, similar to the mixing we observe in the two view case.

\begin{figure}
    \centering
    \includegraphics[scale=0.7]{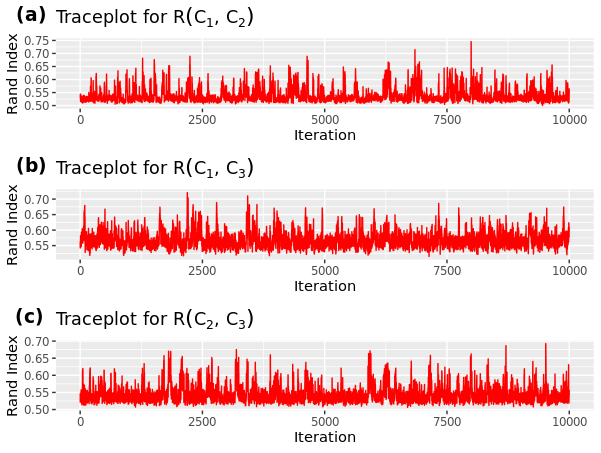}
    \caption{Traceplots for the all pairwise Rand indices for the three partitions in the three view setting.}
    \label{fig:three-view-traces}
\end{figure}

\section{Additional Details on the CPP Application}

The \cite{longnecker2001association} dataset has $d=494$ features measured on $2,379$ mothers (and their children). First, we subset the data to have gestational age strictly less than $42$ weeks in order to remove outliers. Then, we take a random sample of size $n=1,000$ to help with visualization of the results and to add more uncertainty into the cluster-specific parameters, ultimately making the problem more challenging. Both gestational age and birthweight are normalized to have zero mean and unit standard deviation. 

\begin{figure}
    \centering
    \includegraphics[scale=0.7]{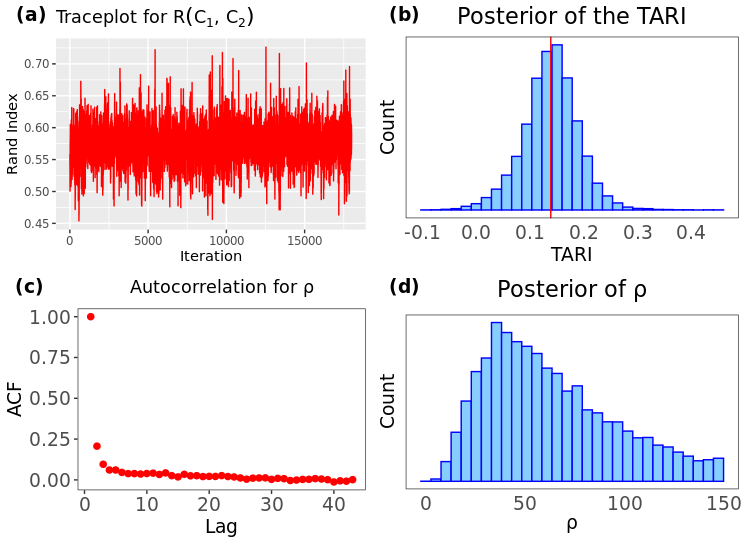}
    \caption{MCMC diagnostics for the application of CLIC to the \cite{longnecker2001association} data.}
    \label{fig:longecker-traces-and-correlations}
\end{figure}

As mentioned in the main article, mixing can improve by using a griddy-Gibbs sampler for $\rho$. We implement the griddy-Gibbs sampler in the application to the \cite{longnecker2001association} dataset, where $\rho$ has support on a grid from $10^{-2}$ to $150$ in intervals of length $0.5$. Figure \ref{fig:longecker-traces-and-correlations} shows the MCMC diagnostic plots for $R(C_1,C_2)$ and $\rho$ over the 18,000 iterations we use. As we observed in simulations, draws in the Gibbs sampler tend to be more correlated for very low values of the Rand index. The effective sample size of $R(C_1,C_2)$ is $1497.218$ and for $\rho$ it is $6251.118$. The time it takes to run the MCMC sampler, compute the point estimates, and sample from the posterior of $R(C_1, C_2)$ was 1,592.387 seconds.

\section{Testing the Uncorrelated and Correlated Models}

In some applications, it may be of interest to test the two multiview models presented in Section \ref{section:PCDP}. In one model, the views are correlated conditional on the cluster labels, and in the other, the views are uncorrelated. One possibility is to compute the Bayes factor (BF) \citep{kass1995bayes} for testing $H_0: f(X_i; \theta_i) = f_1(X_{1i}; \theta_{1i}) \times f_2(X_{2i}; \theta_{2i})$ vs. $H_1: f(X_i; \theta_i) = f_1(X_{1i}; \theta_{1i}) \times f_2(X_{2i}; X_{1i}, \theta_{2i})$, which involves computing the marginal likelihood of both models. In many cases, computing the marginal likelihood may be infeasible. An alternate strategy is to specify a variable selection prior for the base measure in the second view. As a simple example, suppose we have $d_2 = 1$ and set $f_1(X_{1i};\theta_{1i}) = \N(X_{1i};\theta_{1i}, \Sigma_1)$. For the outcome, say we set $\theta_{2i} = (\alpha_{2i}, \beta_{2i})$ and $f_2(X_{2i};X_{1i}, \theta_{2i}) = \N(X_{2i}; \alpha_{2i} + X_{1i}^T\beta_{2i}, \sigma_2^2 )$. To complete the specification, we take the base measure in the second view to be
\begin{equation*}
    H_2(\theta_{2k_2}^*) = \N(\alpha_{2k_2}^*; \mu_\alpha, \Sigma_\alpha) \times \lb \lambda \delta_0(\beta_{2k_2}^*) + (1-\lambda) \N(\beta_{2k_2}^*; \mu_\beta, \Sigma_\beta)   \rb,
\end{equation*}
where $\lambda \in (0,1)$ and $\delta_0$ is a point-mass at the zero vector. The choice of measure for $\beta_{2k_2}^*$ is a version of the spike-and-slab prior for linear regression models \citep{mitchell1988bayesian, ishwaran2005spike}. This is equivalent to mixing over a latent variable $\phi \in \lb 0,1 \rb$, where $(\beta_{2k_2}^* \mid \phi=1) \sim \delta_0$, $(\beta_{2k_2}^* \mid \phi=0) \sim \N(\mu_\beta, \Sigma_\beta)$, and $\pi(\phi=1 \mid \lambda) = \lambda$. When $\phi = 1$, model (5) holds, otherwise, model (6) holds. Suppose further that we set $\lambda \sim \tx{Beta}(a_\lambda, b_\lambda)$. We could test model (5) against model (6) by calculating the BF of $H_0: \phi=1$, or simply $\pi(\phi=1 \mid \bs X)$. This could be accomplished using the output of the Gibbs sampler in Section \ref{section:computation} by saving samples $\phi^{(t)}$ for $t=1, \dots, T$, then applying the approximation $\pi(\phi = 1 \mid \X) \approx (1/T) \sum_{t=1}^T \textbf{1}_{\phi^{(t)} = 1}$.

\section{Alternate Methods of Computation}
In Section \ref{section:CLIC}, we use a result on HDPs from \cite{camerlenghi2018bayesian} to derive the MEPPF for the cross-partition. As with HDPs, the inclusion of $ \tilde q_1$ and $\tilde q_2$ in the Gibbs sampler can be inefficient. However, Theorem \ref{thm:joint-EPPF} shows that we can marginalize out $\tilde q_1, \tilde q_2$ from our sampler. As in Section \ref{section:computation}, we need to sample from the full conditional distributions of $\rho$, $\bs r$, $C_1$, and $C_2$, as well as the cluster-specific parameters. Note that 
\begin{equation}
   \pi(C_1, C_2, \bs r \mid \rho, \gamma_1, \gamma_2) = \frac{\gamma_1^{K_1}\gamma_2^{K_2}}{(\rho)^{(n)}} \rho^{r} \lb \prod_{v=1}^{2}  \frac{\prod_{k_v} (r_{vk_v}-1)!}{(\gamma_v)^{(|r_v|)}} \rb \prod_{k_1,k_2} |s(n_{k_1 k_2}, r_{k_1 k_2})|. 
\end{equation}
As before, we have that $\pi(\rho \mid -) \propto \rho^{r}/(\rho)^{(n)} \pi(\rho)$. The main difficulty is simulating from $\pi(\bs r \mid -) \propto \pi(C_1, C_2, \bs r \mid \rho, \gamma_1, \gamma_2)$. Recall that $\pi(\bs r)$ can be simulated by first sampling $r$ from $\pi(r=w \mid \rho) \propto \rho^w$, then generating the root partitions $T_1$ and $T_2$ from independent $\tx{CRP}(\gamma_v)$ distributions with $r$ customers and taking $\bs r$ to be their contingency table. Therefore, one possible strategy is to take $\pi(\bs r)$ to be a proposal distribution and to perform a Metropolis-Hastings step within the Gibbs sampler. However, when sampling $(T_1, T_2)$ conditional on $(C_1, C_2)$, we have to make sure that the number of clusters in $T_v$ is equal to $K_v$. If the number of clusters in either root partition is not equal to the number of clusters in $C_v$, the candidate would automatically be rejected. Therefore, one needs to be careful for the proposal distribution of $\bs r$ in order to preserve the number of clusters. Alternatively, existing samplers for the HDP could be augmented for PCDP, see \cite{dasblocked2024} and references therein for a comprehensive overview. 

\end{document}